\documentclass[reqno]{amsart}

\usepackage[utf8]{inputenc}
\usepackage{anyfontsize}
\usepackage{stmaryrd}
\usepackage{newtxtext}
\usepackage{newtxmath}
\usepackage{mathtools,amsthm,amsopn}
\usepackage{thm-restate}
\usepackage{amsmath}
\usepackage{soulpos}
\usepackage{enumerate}
\usepackage{bbm}
\usepackage{spectralsequences} 
\usepackage{tikz-cd}
\usepackage[colorinlistoftodos]{todonotes}
\tikzset{
	symbol/.style={
		draw=none,
		every to/.append style={
			edge node={node [sloped, allow upside down, auto=false]{$#1$}}}
	}
}
\usepackage[colorinlistoftodos]{todonotes}
\usepackage[all]{xy}
\usepackage{amsfonts}
\usepackage[bbgreekl]{mathbbol}
\usepackage{csquotes}
\usepackage[
backend=biber,
style=alphabetic,
sorting=nyt,
giveninits=true,
maxcitenames=4,
maxbibnames=4,
maxnames = 4,
maxalphanames=4
]{biblatex}
\renewbibmacro{in:}{}
\usepackage{color}
\usepackage{footmisc}
\usepackage{hyperref}
\usepackage{cleveref}
\declaretheorem[name=Theorem,numberwithin=section]{thm}

\hypersetup{
	colorlinks=true, 
	linktoc=all,     
	linkcolor=blue,  
}
\usepackage{trimclip}
\def\!{\setbox0=\hbox{!}\raisebox{.2\ht0}{\clipbox{0pt .2\ht0 0pt -.1\ht0}{!}}}

\DeclareSymbolFontAlphabet{\mathbbm}{bbold}
\DeclareSymbolFontAlphabet{\mathbb}{AMSb}%
\DeclareFieldFormat[misc]{title}{\mkbibemph{#1}}
\DeclareFieldFormat[article]{title}{\mkbibemph{#1}}
\DeclareFieldFormat[article]{journaltitle}{#1}

\newcommand{\Br}{\mathfrak{Br}}
\newcommand{\hol}{\operatorname{Hol}}
\newcommand{\spec}{\operatorname{Spec}}

\newcommand{\C}{\mathbb{C}}

\newcommand{\uL}{\underline{L}}
\newcommand{\R}{\mathbb{R}}
\newcommand{\Z}{\mathbb{Z}}

\newcommand{\oU}{\operatorname{U}}
\newcommand{\Cx}{\C^\times}

\newcommand{\pt}{\text{pt}}
\newcommand{\pr}{\operatorname{pr}}
\newcommand{\PP}{\mathbb{P}}
\newcommand{\Hom}{\operatorname{Hom}}
\newcommand{\Fuk}{\operatorname{Fuk}}

\newcommand{\im}{\operatorname{Im}}

\newcommand{\lt}{\mathfrak{t}}

\newcommand{\id}{\operatorname{id}}

\newcommand{\fz}{\mathfrak{Z}}

\theoremstyle{plain}
\newtheorem{lem}[thm]{Lemma}
\newtheorem{prop}[thm]{Proposition}
\newtheorem{cor}[thm]{Corollary}

\newtheorem*{thm*}{Theorem}

\theoremstyle{definition}

\newtheorem{df}[thm]{Definition}
\newtheorem{exmp}[thm]{Example}

\theoremstyle{remark}
\newtheorem{rem}[thm]{Remark}

\title{3d Mirror Symmetry is Mirror Symmetry}
\author{Ki Fung Chan and Naichung Conan Leung}
\date{}

\begin{document}
 \begin{abstract}
    3d mirror symmetry is a mysterious duality for certian pairs of hyperk\"ahler
manifolds, or more generally complex symplectic manifolds/stacks. In this paper, we will
describe its relationships with 2d mirror symmetry. This could be regarded
as a 3d analog of the paper 'Mirror symmetry is T-duality' \cite{SYZ} by
Strominger, Yau and Zaslow which described 2d mirror symmetry via 1d dualities.
\end{abstract}
\maketitle

\addtocontents{toc}{\protect\setcounter{tocdepth}{1}}
\section*{Introduction}

From physical considerations, Rozansky and Witten \cite{RW} studied the
3d $\mathcal{N}=4$ supersymmetric $\sigma$-model with target $X$ and its
topological twisted B-model. This 3d TQFT should produce a 2-category of
boundary conditions, as described by Kapustin-Rozansky-Saulina \cite{KRS,KR}.
The A-model is defined only for those $X$ which admit $\operatorname{Sp}(1)$-isometries
rotating the hyperkähler complex structures $I$, $J$, and $K$.

Very loosely speaking, 3d mirror symmetry \cite{IS96} predicts that for certain pairs of
hyperkähler manifolds $X$ and $X^{!}$, possibly stacky, there is a duality
between the 3d A-model on $X$ and the 3d B-model on $X^{!}$, and vice versa,
namely:
\begin{align*}
\text{3d A-model}(X) \quad &\cong \quad \text{3d B-model}(X^{!}), \\
\text{3d B-model}(X) \quad &\cong \quad \text{3d A-model}(X^{!}).
\end{align*}
Typically, $X$ and $X^{!}$ represent the Higgs branch and Coulomb branch of a 3d 
$\mathcal{N}=4$ theory, with their roles got switched in the dual theory.

3d mirror symmetry (also known as symplectic duality) has been extensively studied
from various perspectives (see, e.g., \cite
{ellipticstableenvelopes,BDG,BDGH,BFN,BLPW1,BLPW2,HoloFloer,GHM,IS96,RimanyiBotta,Nakajima}).
\newline

In this article, we aim to understand 3d mirror symmetry as 2d mirror
symmetry for intersections of 3d branes. In the same spirit, SYZ mirror
symmetry \cite{SYZ} understands 2d mirror symmetry as T-duality, which is a
duality in 1d TQFT.

Geometrically a 3d brane on $X$ is a pair $\underline{C}=(C,Y)$, where $C$
is a complex Lagrangian of $X$, and $Y$ is a K\"{a}hler manifold with a
holomorphic map $\pi:Y\rightarrow C$. Sometimes we write $\underline{C}:Y\overset{%
\pi }{\rightarrow }C\subset X$ and $Y_{c}=\pi ^{-1}\left( c\right) $. If $%
\pi $ is a fiber bundle, then (i) $\left\{ \Fuk\left( Y_{c}\right) \right\}
_{c\in C}$ forms a local system of categories over $C$ and (ii) $\left\{
D^{b}\left( Y_{c}\right) \right\} _{c\in C}$ forms a holomorphic family of
categories over $C$. A 3d brane can be viewed as a boundary condition for
both the A-model and B-model of $X$, and we will refer to it as a 3d A-brane
or 3d B-brane accordingly. In this paper, we also consider more general
branes where $Y$ is an Landau-Ginzburg model (LG model for short).

Given 3d branes $\underline{C}_{1}=(C_{1},Y_{1})$ and $\underline{C}%
_{2}=(C_{2},Y_{2})$ on $X$, if $C_{1}$ and $C_{2}$ intersect transversely,
say at a single point $p\in X$, then their intersection $\underline{C}%
_{1}\cap \underline{C}_{2}$ is the K\"ahler manifold $Y_{1,p}\times Y_{2,p}$,
which determines a 2d TQFT in either A-model or B-model. We denote $%
_{A}\Hom\left( \underline{C}_{1},\underline{C}_{2}\right) =\Fuk\left(
Y_{1,p}\times Y_{2,p}\right) $ and $_{B}\Hom\left( \underline{C}_{1},%
\underline{C}_{2}\right) =D^{b}\left( Y_{1,p}\times Y_{2,p}\right) $. In
general, we expect that, in either A-model or B-model, 3d branes on $X$
generate a 2-category with hom categories are given by intersection theory
of 3d branes (see \cite{GHM} for the proof of a version of  categorical 3d mirror symmetry).

When $X$ and $X^{!}$ are 3d mirror to each other, then there should be
correspondences for certain 3d branes $\underline{C}$'s in $X$ and $%
\underline{C}^{!}$'s in $X^{!}$, so that the 2d TQFTs $_{A}\Hom_X\left( 
\underline{C}_{1},\underline{C}_{2}\right) $ and $_{B}\Hom_{X^!}\left( \underline{C}%
_{1}^{!},\underline{C}_{2}^{!}\right) $ should be 2d mirror to each other.
The possible choices of $C$'s and $C^{!}$'s are dictated by a choice of a
K\"ahler parameter and an equivariant parameter on $X$ and also on $X^{!}$,
which got interchange under 3d mirror symmetry. In particular, the 3d mirror
symmetry for $X=X^{!}=\pt$ is simply mirror symmetry.

To turn the tables around, functoriality of mirror symmetry would greatly
constraint the configuration of these complex Lagrangians in $X^{!}$, given
the configuration of their corresponding complex Lagrangians in $X$. For
example this would allow us to construct the 3d mirror to hypertoric
varieties via a gluing method in \Cref{Gluing}. We view this as a 3d analog of the SYZ
construction of mirror manifolds. 

The simplest examples of 3d mirror pairs are Gale dual hypertoric varieties of
the form $X=T^{\ast }\mathbb{C}^{n}\sslash T_{\C}$, and $X^{!}=T^{\ast }\mathbb{C}%
^{n}\sslash \check{T}_{\C}^{\prime }$, where 
\begin{align*}
1\rightarrow T_{\C}\rightarrow & (\mathbb{C}^{\times })^{n}\rightarrow
T_{\C}^{\prime }\rightarrow 1 \\
1\rightarrow \check{T}_{\C}^{\prime }\rightarrow & (\mathbb{C}^{\times
})^{n}\rightarrow \check{T}_{\C}\rightarrow 1
\end{align*}%
are dual exact sequences of complex tori. The correspondences between 3d
branes in $X$ and $X^{!}$ provide insights into 2d mirror symmetry. Some
examples are provided below (we allow $T_{\C}\rightarrow (\mathbb{C}^{\times
})^{n}$ to be non-injective, see \Cref{Functoriality1}).

\begin{center}
\begin{tabular}{|c|c|c|c|c|}
\hline
\textbf{$X$} & \text{$\mathfrak{Br}_{A}(X)$} & \text{$X^{!}$} & \text{$\mathfrak{Br}_{B}(X^{!})$} & \text{2d Interpretations} \\ 
\hline
$[\text{pt}\sslash \mathbb{C}^{\times}]$ & 
$\begin{array}{c}
S^1 \curvearrowright Y \\ 
\text{Equivar. structure}
\end{array}$ & 
$\ T^{\ast}\mathbb{C}^{\times}$ & 
$\begin{array}{c}
Y^{\vee} \rightarrow \mathbb{C}^{\times} \\ 
\text{Fibration over } \mathbb{C}^{\times}
\end{array}$ & 
$\begin{array}{c}
\text{Teleman Conjecture} \\ 
\text{\cite{tel2014}}
\end{array}$ \\ 
\hline

$[T^{\ast}\mathbb{C}\sslash \mathbb{C}^{\times}]$ & 
$\begin{array}{c}
(Y, D) \\ 
\text{log CY}
\end{array}$ & 
$T^{\ast}\mathbb{C}$ & 
$\begin{array}{c}
Y^\vee\to \C \\ 
\text{LG Model}
\end{array}$ & 
$\begin{array}{c}
\text{MS for log CY} \\ 
\text{\cite{AurouxComp}}
\end{array}$ \\ 
\hline

$[T^{\ast}\mathbb{C}_{1,-1}^{2}\sslash \mathbb{C}^{\times}]$ & 
$\begin{array}{c}
Y \rightarrow \mathbb{C} \\ 
\text{Tyurin Degen.}
\end{array}$ & 
$T^{\ast}\mathbb{P}^{1}$ & 
$\begin{array}{c}
Y^{\vee} \rightarrow \mathbb{P}^{1} \\ 
\text{Fibration over } \mathbb{P}^{1}
\end{array}$ & 
$\begin{array}{c}
\text{DHT Conjecture} \\ 
\text{\cite{DHT}}
\end{array}$ \\ 
\hline
\end{tabular}
\end{center}

The brackets in the $X$ indicates that the stacky structure of $X$ plays an
important role. More detailed explanations of these examples can be found in %
\Cref{relatedworks}. The general transform for the Gale dual case is
described in \Cref{2.1,2.2}, though
most of the proofs are deferred to \Cref{Aurouxpicture}.

\subsection*{The exchange of K\"ahler and equivariant parameters}
Let $\text{Ham}(X)$ be the group of Hamiltonian automorphisms of $X$, and $\mathcal{K}_X$ be the semigroup of complexified K\"ahler classes of $X$. There is an action (see \Cref{matchingKahEqu})
\begin{equation}\label{Inaction}
    \text{Ham}(X)\times \mathcal{K}_X\curvearrowright \Br(X)
\end{equation}
of $\text{Ham}(X)\times \mathcal{K}_X$ on the set $\Br(X)$ of all 3d branes on $X$.

The action of $\text{Ham}(X)$ deforms the complex structures of 3d branes, while the action of $\mathcal{K}_X$ deforms their symplectic structures. Since complex structures and symplectic structures are swapped under 2d mirror symmetry, it is natural to expect the action of $\text{Ham}(X)$ (resp. $\mathcal{K}_X$) on $\Br(X)$ and that of $\mathcal{K}_{X^!}$ (resp. $\text{Ham}(X^!)$) on $\Br(X^!)$ are related by 3d mirror symmetry.

This is tied to the exchange of K\"ahler and equivariant parameters in 3d mirror symmetry as follows. The space of K\"ahler parameters of $X$ is $H^{2}(X;\mathbb{C})$, while the space of equivariant parameters is the Lie algebra $\mathfrak{t}_{X,\mathbb{C}}$ of a maximal torus of $\text{Ham}(X)$. The action \eqref{Inaction} induces an infinitesimal action of $\mathfrak{t}_{X,\mathbb{C}} \times H^2(X;\mathbb{C})$ on $\Br(X)$.

Under the transforms of 3d branes, we expect the infinitesimal action of $\mathfrak{t}_{X,\mathbb{C}}$ (resp. $H^2(X;\mathbb{C})$) on $\Br(X)$ and that of $\mathcal{K}_{X^!}$ (resp. $\mathfrak{t}_{X^!,\mathbb{C}}$) on $\Br(X^!)$ are interchanged, explaining the isomorphisms:
\begin{equation*}
\begin{array}{cccc}
\Phi:&H^{2}(X;\mathbb{C}) & \simeq  & \mathfrak{t}_{X^{!},\mathbb{C}} \\ 
\Psi:&\mathfrak{t}_{X,\mathbb{C}} & \simeq  & H^{2}(X^{!};\mathbb{C}).
\end{array}
\end{equation*}

In other words, the exchange of K\"ahler and equivariant parameters in 3d mirror symmetry corresponds to the exchange of symplectic and complex structures in 2d mirror symmetry\footnote{A similar conjecture was formulated in \cite{BDGH}}.
\newline

Let $X$ be an exact complex symplectic manifold/stack, and let $X^!$ denote its 3d mirror. We propose the following:
\begin{enumerate}[I)]
    \item Mirror 3d branes of $X$ and $X^!$ are related by an SYZ-type transform.
    \item The exchange of K\"ahler and equivariant parameters between $X$ and $X^!$ is compatible with the exchange of symplectic and complex structures in 2d mirror symmetry via the action \eqref{Inaction}.
    \item If $X$ is obtained by gluing complex symplectic stacks $U_i$'s, then $X^!$ can be obtained by gluing the mirror complex symplectic stacks $U_i^!$'s.
\end{enumerate}

\subsubsection*{Explanations for I): the Abelian case}
If $X$ is a hypertoric variety/stack (the Abelian case), and $C$ is one of the core Lagrangians of $X$, we develop an SYZ-type transform from 3d branes supported on $C$ to 3d branes in $T^*C^!$, where $C^!$ is the Gale dual of $C$ (see \Cref{galedef}). The main theorem is as follows:

\begin{thm}[See \Cref{SYZtransform} for the precise version]
    Let $Y \to C$ be a 3d brane supported on $C$, and let $Y^\vee$ be an SYZ-mirror of $Y$ constructed in \Cref{Aurouxpicture}. Then there exists a 3d brane $[Y^\vee \sslash \check{T}_{X}] \to C^!$, supported on $C^!$.
\end{thm}

In this construction, $Y^\vee$ is the moduli space of Lagrangian branes in the complement of an anticanonical divisor in $Y$, and the map $[Y^\vee \sslash T_{X}] \to C$ is induced by (part of) the LG potential. 

In many situations (e.g. $X=T^*[\C^n/T_\C]$, $C=[\C^n/T_\C]$), $T^*C^!$ would be an open substack of $X^!$, thus $[Y^\vee \sslash T_{X}] \to C$ can be viewed as a 3d brane on $X^!$. In general, $T^*C^!$ and $X^!$ are related by a complex Lagrangian correspondence, which transforms $[Y^\vee \sslash T_X] \to C$ into a 3d brane on $X^!$ (see \Cref{stratified}, and the explanations for III) below). For simplicity of notation, we assume in the rest of this introduction that $T^*C^!$ is an open substack of $X^!$.

\subsubsection*{Explanations for I): the non-Abelian case}

Let $G$ be a compact group, and $T$ a maximal torus of $G$. The transform of 3d branes between the Higgs and Coulomb branches of a 3d $\mathcal{N}=4$ $G$-gauge theory (\cite{BDG,BDGH,Nakajima}) is induced by the transform of 3d branes between the Higgs and Coulomb branches of the corresponding $T$-gauge theory.

If $X$ is the Higgs branch of a 3d $\mathcal{N}=4$ $G$-gauge theory, and we denote by $X_T$ the Higgs branch of the corresponding $T$-gauge theory (i.e., the restriction of the gauge group), then there exists a birational morphism $X^! \to (X_T^!)/W_G$, where $W_G$ is the Weyl group of $G$. 

The SYZ-type transform of 3d branes in this case can be described through the process of Abelianization. A 3d brane of $X$ can also be regarded as a 3d brane of $X_T$. We propose that the corresponding mirror 3d brane in $X_T^!$ is $W_G$-invariant, and its quotient in $(X_T^!)/W_G$ can be lifted to $X^!$.

The method of Abelianization requires proving some nontrivial theorems in Floer theory. An example is provided in the following theorem:
\begin{thm}[See \Cref{Floerthm1}]
Let $Y$ be a $G$-Hamiltonian manifold, and let $L \subset Y$ be a compact, connected, $T$-relatively spin, Lagrangian submanifold, with $E$ a unitary flat line bundle on $L$. If $HF^\bullet((L,E),(L,E)) \neq 0$, then $\fz(L,E) \triangleq \exp(-\mu_T(L)) \hol_T(E) \in T_\C^\vee$ lies in the center of the Langlands dual group $\check{G}$.
\end{thm}

We explain the connection between this theorem (as well as its generalizations) and the 3d brane transforms in the non-Abelian case in \Cref{nonabelian}. The proofs of these "Abelianization" theorems can be found in our earlier work \cite{paper2}.

\subsubsection*{Explanations for II)}

As mentioned above, when we vary the K\"ahler parameters (resp. equivariant parameters) of $X$, it deforms the symplectic structures (resp. complex structures) of its 3d branes (and of their intersections). Thus, the exchange of K\"ahler and equivariant parameters in 3d mirror symmetry corresponds to the exchange of symplectic and complex structures in 2d mirror symmetry.

Our SYZ-type transform of 3d brnaes can be used to verify this phenomenon, as described in the following theorem. We refer the readers to \Cref{matchingKahEqu} for more detailed explanations and examples.

\begin{thm}[See \Cref{thmparameters} for the precise version]
Let $\pi:Y \to C$ be a 3d brane in $X$ supported on $C$, and let $F:[Y^\vee \sslash T_X] \to C^!$ denote its mirror 3d brane in $X^!$. Then:
    \begin{enumerate}
        \item Let $\alpha \in H^2(X,\C)$, and assume that the class $[\omega_Y] + \pi^*\alpha$ is symplectic on $Y$. Denote by $Y_\alpha$ the space $Y$ equipped with the symplectic class $[\omega_Y + \pi^*\alpha]$. Then, the 3d brane $\pi:Y_\alpha \to C$ is mirror to the 3d brane $F_\alpha:[Y^\vee \sslash T_X] \to C^!$, where
        \[
        F_\alpha = \exp(-\Phi(\alpha)) F.
        \]
        \item Let $\xi \in \mathfrak{t}_{X,\mathbb{C}}$. Then, the 3d brane $\pi_\xi = \exp(\xi) \cdot \pi: Y \to C$ is mirror to the 3d brane $F: [Y^\vee \sslash T_X]_\xi \to C$, where $[Y^\vee \sslash T_X]_\xi$ is $[Y^\vee \sslash T_X]$ equipped with the symplectic class
        \[
        [\omega_{[Y^\vee \sslash T_X]}] + F^* \Psi(\xi).
        \]
    \end{enumerate}
\end{thm}

\subsubsection*{Explanations for III)}
Consider the example $X=T^*\mathbb{P}^1$ and $X^!=T^*[\C^2_{1,-1}/\Cx]$. We have 
\[
X=T^*C=T^*C_1\bigcup_{T^*C_{12}}T^*C_2,
\]
where $C=\mathbb{P}^1$, $C_1\cong \C\cong C_2$, and $C_{12}\cong \Cx$. On the other hand, $X^!=T^*C^!$ does not contain $T^*C_1^!$ and $T^*C_2^!$ as open substacks. Instead, we have $C^!=C^!_1\times _{C_{12}^!}C^!_2$, and $X^!$ is related to (and can be built from) $T^*C_1^!$, $T^*C_2^!$ and $T^*C_{12}^!$ via complex Lagrangian correspondence. This will be discussed in \Cref{SYZtype}.

On the other hand, let $V$ be a $T$-representation, and consider $X=[T^*V\sslash T_\C]$. There are natural complex Lagrangian correspondences (see \Cref{glueing section}) from $T^*[\pt/T_\C]$ to $X$, each of which is expected to be mirror to an open inclusion $\mathcal{U}_I\cong T^*\check T_\C\to X$. By using the SYZ-type transform and the complex symplectic groupoid structure on $T^*\check T_\C$, we derive a formula for the change of coordinates
\[
\phi_{J,I}:\mathcal{U}_I\to\mathcal{U}_J.
\]
We refer the readers to \Cref{Gluingintro} for some examples. We verify III) by checking the cocycle conditions of $\phi_{J,I}$. The following is the main result of \Cref{Gluing}.
\begin{thm}[See \Cref{gluingthm} and \Cref{thm4II}]
    $\phi_{J,I}$ satisfies the cocycle conditions. Moreover, we can obtain the BFN Coulomb branch described in \cite{BFN} by gluing $T^*\check T_\C$ using $\phi_{J,I}$ as an affine scheme.  
\end{thm}

\subsection*{Organization of the paper}

In \Cref{Functoriality1}, after a brief review of 2d mirror symmetry, we
describe the SYZ type transform of boundary conditions from a functorial
perspective. We then provide a second introduction to our approach through
examples.

Proofs of the SYZ type transform are provided in \Cref{Aurouxpicture}. In \Cref{matchingKahEqu}, the SYZ type transform was used to explain the exchange of K\"ahler and equivariant parameters in 3d mirror symmetry.

In \Cref{Gluing}, we use the functoriality of 2d mirror symmetry to derive a gluing formula for the Coulomb branch of a gauge theory. Our formula differs slightly from the one proposed by Teleman \cite{tel2021}.

 \addtocontents{toc}{\protect\setcounter{tocdepth}{0}}
\section*{Acknowledgements}
The authors thank Mohammed Abouzaid, Kwokwai Chan, Tudor Dimofte, Justin Hilburn, Chin Hang Eddie Lam, Siu-Cheong Lau, Yan-Lung Leon Li, Sukjoo Lee, Ziming Nikolas Ma, Cheuk Yu Mak, Michael McBreen, and Yat Hin Suen for valuable discussions on various stages of this project. The work of N. C. Leung described in this paper was substantially supported by grants from the Research Grants Council of the Hong Kong Special Administrative Region, China (Project No. CUHK14301721, CUHK14306322, and CUHK14305923) and a direct grant from CUHK.

\addtocontents{toc}{\protect\setcounter{tocdepth}{2}}
\section{3d Mirror Symmetry is mirror symmetry}

\label{Functoriality1}

\subsection{SYZ mirror symmetry, a review}

In string theory, given a Calabi-Yau manifold $Y$, we can associate a 2d $\mathcal{N}=(2,2)$ supersymmetric $
\sigma $-model with target $Y$ \cite{MS}. It admits two topological twists,
thus giving 2d A-model of the symplectic geometry of $Y$ and 2d B-model of
the complex geometry of $Y$.

For a large class of Calabi-Yau manifolds $Y$ (or more generally K\"ahler
manifolds or LG-models), string theory predicts that there is another
Calabi-Yau manifold $Y^{\vee }$ such that 
\begin{eqnarray*}
\text{2d A-model}\left( Y\right) &\cong &\text{2d B-model}\left( Y^{\vee
}\right) , \\
\text{2d B-model}\left( Y\right) &\cong &\text{2d A-model}\left( Y^{\vee
}\right) .
\end{eqnarray*}%
This is the 2d mirror symmetry, or simply mirror symmetry (see e.g. \cite{MS}). In open string sectors,
this duality identifies the categories of boundary conditions, or branes, which are objects of the (derived) Fukaya category $\Fuk(Y)$ of Lagrangians in A-model, or that of the bounded derived category $D^b(Y)$ of coherent sheaves in B-model. More precisely, it is predicted that there are equivalence of categroies
\begin{eqnarray*}
\Fuk\left( Y\right) &\cong &D^b\left( Y^{\vee }\right) , \\
D^b\left( Y\right) &\cong &\Fuk\left( Y^{\vee }\right) .
\end{eqnarray*}%
This is the Homological Mirror Symmetry (abbrev. HMS) conjecture by
Kontsevich \cite{HMS}.

A geometric boundary condition in $Y$ is given by a submanifold $L \subset Y$, coupled with a family of 1d  supersymmetric $\sigma$-models parametrized by $L$. Quantization of the supersymmetric 1d $\sigma$-model, or quantum mechanics, on a manifold $M$, is given by the Hilbert space $H^{\ast}(M,\mathbb{C})$, as shown by Witten \cite{SUSYMorse}. In other words, a brane is a Hermitian vector bundle $E$ over a submanifold $L$ in $Y$. From SUSY considerations, we require $E$ to satisfy:
(i) $E$ is a flat bundle over a Lagrangian submanifold of $Y$ in the A-model; or
(ii) $E$ is a holomorphic vector bundle over a complex submanifold of $Y$ in the B-model.

\citeauthor{SYZ} \Cite{SYZ} proposed a ground breaking geometric explanation
of the 2d mirror symmetry in terms of T-duality, which can be reviewed as a duality in 1d SUSY quantum mechanics. Loosely speaking, it says that if $Y$ and $Y^{\vee }$
is a mirror pair of CY manifolds, then (i) they possess dual (special)
Lagrangian torus fibrations, at least near large volume/complex structure
limits; (ii) fiberwise Fourier-Mukai transformation exchanges the symplectic
geometry of $Y$ with the complex geometry of $Y^{\vee }$ and vice versa.

In particular, $Y^{\vee }$ is a (compactification of) moduli space of
Lagrangian tori $L$ in $Y$ coupled with unitary flat line bundle $E$ on $L$.%
\begin{equation*}
Y^{\vee }=\left\{ \left( L,E\right) :\text{Lag. torus }L\subset Y\text{, }\ 
\text{ flat } \oU(1)\text{-bundle } E\text{ over }L\right\} /\cong \text{.}
\end{equation*}%
Much of the discussion below can be carried out without the topological
constraint on $L$ being a torus. A precise construction of $Y^{\vee }$ is
rather technical as it requires nontrivial quantum corrections and wall-crossing (see e.g. 
\cite{KontSoi1,KontSoi2}).

Mirror symmetry conjecture has also been generalized as a duality between
symplectic geometry and complex geometry in other settings, most notably for
Fano manifolds $Y$ and LG models $\left( Y^{\vee
},W\right) $ (see e.g \cite{AurouxComp}). When $Y$ is a Fano manifold,
namely $c_{1}\left( Y\right) >0$, and $D\in \left\vert -K_{Y}\right\vert $
is a simple normal crossing anti-canonical divisor of $Y$, we consider the
moduli space $Y^{\vee }=\left\{ \left( L,E\right) \right\} /\cong $ as
before with $L$ being a Lagrangian submanifold inside the log-CY manifold $%
Y\backslash D$. Counting holomorphic disks bounding $L$ in $Y$ and meeting $D
$ at a point, weighted by holonomy of $E$ along disk boundaries, defines a
holomorphic function 
\begin{equation*}
W:Y^{\vee }\rightarrow \mathbb{C}\text{,}
\end{equation*}
called the potential function of the LG model $\left( Y^{\vee },W\right) $.
HMS conjecture says that 
\begin{eqnarray*}
\Fuk\left( Y\right) &\cong &D_{sing}\left( Y^{\vee },W\right) , \\
D^{b}\left( Y\right) &\cong &FS\left( Y^{\vee },W\right) ,
\end{eqnarray*}
where $D_{sing}\left( Y^{\vee },W\right) $ is the triangulated category of
singularities (\cite{Orlov}) and $FS\left( Y^{\vee },W\right) $ is the
Fukaya-Seidel category (\cite{Seidelbook}).

\subsection{Branes in 3d theory}

As mentioned in the introduction, the spaces we study are complex symplectic
manifolds/stacks. We refer the readers to \cite{ShiftSym} for general
symplectic stacks.

Let $G_\C$ be a complex reductive Lie group with a Hamiltonian action on a complex
symplectic manifold $\mathcal{X}$ with moment map $\mu_\C:\mathcal{X}\to 
\mathfrak{g}^*_\C$, then 
\begin{equation*}
X=[\mathcal{X}\sslash G_\C]=[\mu_\C^{-1}(0)/G_\C]
\end{equation*}
is a complex symplectic stack. All the examples of complex symplectic stacks we considered are of the above form. When $\mathcal{X}$ is the cotangent bundle of a smooth complex $G_\C$-manifold $M$, then we always assume $\mathcal{X}$ is
equipped with the standard $G_\C$-moment map. In particular, we have 
\begin{equation*}
[\mathcal{X}\sslash G_\C]\cong T^*[M/G_\C]
\end{equation*}
if either (i) the $G_\C$-action on $M$ is free, or (ii) if we understand
both the cotangent bundles and the moment map levels in the realm of derived
geometry.

Let $X=[\mathcal{X}\sslash G_\C]$, and $M\subset \mu_\C^{-1}(0)$ be a complex
Lagrangian of $\mathcal{X}$, then $C=[M/G_\C]$ is a $G_\C$-invariant complex Lagrangian of $X$. In the following, we fix a compact real form $G$ of $G_\C$.

\begin{df}
Let $C=[M/G_\C]$ be a complex Lagrangian of $X$, then a 3d brane of $X$
supported on $C$ is a $G$-Hamiltonian K\"ahler manifold $Y$ together with a $%
G$-equivariant holomorphic map $Y\to M$. We denote it as $([M/G_\C],[Y\sslash G])$ or $[Y\sslash G]\to[M/G_\C]$. If the $G$-action on $Y$ extends to a $G_\C$-action, then we also write $([M/G_\C],[Y/G_\C])$, or $%
[Y/G_\C]\to[M/G_\C]$.
\end{df}
In most interesting cases, the $G$-action on $Y$ can extends to a $G_\C$-action, and the symplectic reduction $Y\sslash G$ can be identified with a GIT quotient $Y/G_\C$.

A 3d brane can be viewed as a boundary condition for the A-model or B-model
of $X$, and is accordingly referred to as a 3d A-brane or 3d B-brane.

We denote by 
\begin{equation*}
\mathfrak{Br}(X)
\end{equation*}
the collection of all 3d branes of $X$, and
\begin{equation*}
\mathfrak{Br}(X;C)
\end{equation*}
the collection of all 3d branes of $X$ supported on $C$. We use the notations 
$\mathfrak{Br}_A(X)$, $\mathfrak{Br}_B(X;C)$, etc. when it is necessary to emphasize the specific models.

We will now describe SYZ-transforms in the 3d branes level of hypertoric 3d
mirror symmetry. These transforms are all understood as from A-models to
B-models.

\subsection{3d Mirror symmetry of \texorpdfstring{$T^*[\pt/T_\C]$}{TEXT} and \texorpdfstring{$T^*\check T_\C$}{TEXT}}\label{2.1}
In this subsection, we explain 3d mirror symmetry between the A-model of $X=[\pt\sslash T_\C]=T^*[\pt/T_\C]$, and the B-model of $X^!=T^*\check T_\C$. Pictorially,
\begin{center}
 \begin{tabular}{ c c c }\\
 \begin{tabular}{ c }
 A-model \\
 $T^*[\pt/T_\C]$
\end{tabular} 
& $\underleftrightarrow{\text{ 3d mirror symmetry }}$ & 
 \begin{tabular}{ c }
 B-model \\
 $T^*\check T_\C$
\end{tabular} \\
\\
\end{tabular}
\end{center}
This correspondence (and its converse direction) has first been observed in \cite{tel2014}.

A 3d A-brane of $T^*[\pt/T_\C]$ supported on $[\pt/T_\C]$ is a Hamiltonian $T$-manifold $Y$; while a 3d brane of $T^*\check T_\C$ supported on $\check T_\C$ is a holomorphic function $\pi^\vee:Z\to \check T_\C$ from a K\"ahler manifold $Z$. The 3d mirror symmetry in this case can be explained by taking $Z=Y^\vee$ to be a 2d mirror of $Y$, and the existence of $\pi^\vee$ (with further properties) is asserted by Teleman \cite{tel2014}. In the following, $\pi^\vee$ is called the Teleman map.

To see why $Z=Y^\vee$, we first recall that duality in 1d is the T-duality between $T$ and its dual $\check{T}$:
    \[
    T \xleftrightarrow{\quad \text{T-duality} \quad} \check{T}.
    \]
    In 2d, The simplest form of SYZ mirror symmetry is the duality between $T_{\mathbb{C}}$ and $\check{T}_{\mathbb{C}}$. They possess special Lagrangian fibrations over $\mathfrak{t} \cong \check{\mathfrak{t}}$, with fibers $T$ and $\check{T}$:
    \[
    T_{\mathbb{C}} \xleftrightarrow{\quad\text{ SYZ } \quad} \check{T}_{\mathbb{C}}.
    \]
    If we have a homomorphism
    \[
    T_{\mathbb{C}} \to T'_{\mathbb{C}},
    \]
    then the pair $\left( [\pt/T_{\mathbb{C}}], [T'_{\mathbb{C}}/T_{\mathbb{C}}] \right)$ defines a 3d brane in $T^*[\pt/T_{\mathbb{C}}]$. On the other hand, the dual homomorphism
    \[
    \check{T}'_{\mathbb{C}} \to \check T_{\mathbb{C}}
    \]
    defines a 3d brane in $T^*\check{T}_{\mathbb{C}}$.

    Our basic 3d mirror transform asserts that these two 3d branes are 3d mirror to each other. In particular, one should be able to see the equivalence of their hom categories via 2d mirror symmetry.
    \begin{align*}
    &_A \Hom\left( \left( [\pt/T_{\mathbb{C}}], [T'_{\mathbb{C}}/T_{\mathbb{C}}] \right), \left( [\pt/T_{\mathbb{C}}], [T_{\mathbb{C}}/T_{\mathbb{C}}] \right) \right) \\
    =& \Fuk\left( [T'_{\mathbb{C}}/T_{\mathbb{C}}] \times_{[\pt/T_{\mathbb{C}}]} [T_{\mathbb{C}}/T_{\mathbb{C}} ]\right) = \Fuk(T'_{\mathbb{C}}),
    \end{align*}
    and
    \begin{align*}
    &_B \Hom\left( (\check{T}_{\mathbb{C}}, \check{T}'_{\mathbb{C}}), (\check{T}_{\mathbb{C}}, \check{T}_{\mathbb{C}}) \right) \\
    =& D^b\left( \check{T}'_{\mathbb{C}} \times_{\check{T}_{\mathbb{C}}} \check{T}_{\mathbb{C}} \right) = D^b(\check{T}'_{\mathbb{C}}).
    \end{align*}

    More generally, if the 3d branes $\left( [\pt/T_{\mathbb{C}}], [Y/T_{\mathbb{C}}] \right)$ in $T^*[\pt/T_{\mathbb{C}}]$ and $(\check T_{\mathbb{C}}, Z)$ in $T^*\check T_\C$ are mirror to each other, then the hom categories
    \[
    _A \Hom\left( \left( [\pt/T_{\mathbb{C}}], [Y_{\mathbb{C}}/T_{\mathbb{C}}] \right), \left( [\pt/T_{\mathbb{C}}], [T_{\mathbb{C}}/T_{\mathbb{C}}] \right) \right) = \Fuk(Y),
    \]
    \[
    _B \Hom\left( (\check{T}_{\mathbb{C}}, Z), (\check{T}_{\mathbb{C}}, \check{T}_{\mathbb{C}}) \right) = D^b(Z)
    \]
    are equivalent. Hence, $Y$ and $Z$ are 2d mirror to each other. This is summarized in the following diagram.
    \begin{equation*}
\begin{tikzcd}
    &\left[Y\sslash T\right]\ar[d]\\
    T^*\left[\pt/T_\C\right]&\left[\pt/T_\C\right]\ar[l,symbol=\subset]
\end{tikzcd}
\text{ is 3d mirror to }
\begin{tikzcd}
    Z=Y^\vee \ar[d]&\\
    \check T_\C\ar[r,symbol=\subset]&T^*\check T_\C.
\end{tikzcd}
\end{equation*}%

Suppose $\left( Y,\omega \right) $ is a Hamiltonian $T$-manifold with moment map
\begin{equation*}
\mu=\mu_Y :Y\rightarrow \mathfrak{t}^{\ast }\text{.} 
\end{equation*}
The Teleman map $\pi^\vee$ can be given explicitly using SYZ symmetry. Given any $T$-invariant Lagrangian submanifold $L$ in $Y$, $\mu |_{L}$ is always a constant function, we denote its value as $\mu \left( L\right) \in 
\mathfrak{t}^{\ast }$. Furthermore, up to Hamiltonian equivalences, Lagrangian deformations of $L$ are also $T$-invariant. This reflects the
fact the action of $T$ on $H^{1}\left( L\right) $ is always trivial. Suppose $Y^\vee$ can be constructed as an SYZ mirror, namely a moduli space of certain pairs $(L,E)$ of $T$-invariant Lagrangian torus $L$ and unitary flat line bundle $E$ on $L$, then $\pi^\vee$ can be given as
\begin{eqnarray*}
\pi^\vee=\pi^\vee_Y &:&Y^{\vee }\rightarrow \check T_\C=\check T\times \mathfrak{t}^{\ast } \\
\pi^\vee\left( L,E\right) &=&\left( \hol_T(E) ,-\mu \left( L\right)
\right) .
\end{eqnarray*}%
Here $\hol_{T}\left( E\right) $ is the holonomy of $E$ along the $T$-orbits in $L$. 
\begin{exmp}\label{equi2d}
    \begin{enumerate}[(a)]
        \item []
        \item When $S^{1}\curvearrowright Y=\mathbb{C}^{n}\overset{\mu }{%
\rightarrow }\mathbb{R}$ with $\mu \left( z_{1},\cdots ,z_{n}\right) =\left(
\left\vert z_{1}\right\vert ^{2}+\cdots +\left\vert z_{n}\right\vert
^{2}\right) /2$, we have $\left( Y^{\vee },W\right) =\left( \left( \mathbb{C}%
^{\times }\right) ^{n},z^{1}+\cdots +z^{n}\right)$ and $\pi^\vee=z^{1}\cdots z^{n}$.
    \item When $S^{1}\curvearrowright Y=\pt\overset{\mu }{\rightarrow }%
\mathbb{R}$ with $\mu \left( \pt\right) =\lambda $, we have $\left( Y^{\vee
},W\right) =\left( \pt,0\right) $ and $\pi^\vee=e^{-\lambda }$.
\item When $T^{n}\curvearrowright Y=\PP^{n}\overset{\mu }%
{\rightarrow }\mathbb{R}^{n}$, we have $\left( Y^{\vee },W\right) =\left(
\left( \mathbb{C}^{\times }\right) ^{n},z^{1}+\cdots +z^{n}+\frac{1}{%
z^{1}\cdots z^{n}}\right) $ and $\pi^\vee=\id$. In general $\pi^\vee=\id$ in the toric cases.
    \end{enumerate}
\end{exmp}

\begin{rem}
    We do not need $L$ to be a torus so far, but it is necessary latter when we discuss Hamiltonian actions on $Y^\vee$.
\end{rem}

\begin{rem}\label{HomlevelTeleman}
\begin{enumerate}[(a)]
    \item []
    \item Let $Y'=\pt$ be with the trivial $T$-action and constant moment map with value $\lambda$. Then we can intrepret the fiber product $[\bar Y'\sslash T]\times_{[\pt/T_\C]}[Y/T_\C]$ as $Y\sslash _\lambda T=\mu_Y^{-1}(\lambda)/T$. On the other hand, the fiber product $Y'^\vee\times_{\check T_\C}Y^\vee$ is equal to $(\pi^\vee)^{-1}(e^{-\lambda})$. The hom level statement for 3d mirror symmetry says 
    \[\Fuk(Y\sslash _\lambda T)\simeq D^b((\pi^\vee)^{-1}(e^{-\lambda})).\] 
    This is also a conjecture of Teleman \cite{tel2014}, which says that the symplectic quotients of $Y$ are 2d mirror to the fibres of $\pi^\vee:Y^\vee\to \check T_\C$. In \Cref{TelConj}, we checked this Teleman conjecture locally in the space level using SYZ's picture. The harder part of Teleman conjecture, which involves the superpotential, is tackled in the joint work \cite{LLL} of the second named author with Lau and Li.
    \item Let $Y_1,Y_2$ be two $T$-Hamiltonian manifolds, and $\pi_1:Y_1^\vee\to \check T_\C$, $\pi_2:Y_2^\vee\to \check T_\C$ be the corresponding Teleman maps. We can interpret 
    \[[\bar{Y}_1\sslash T]\times_{[\pt/T_\C]} [Y_2\sslash T]
    \]
    as $(\bar Y_1\times Y_2)\sslash T$ (where $\bar{Y}$ is $Y$ with the opposite symplectic form and moment map). On the other hand, $\bar Y_1\times Y_2$ with the diagonal $T$-action gives the Teleman map $\pi^\vee=\pi_2^\vee/\pi_1^\vee:Y_1^\vee\times Y_2^\vee\to \check T_\C$. Teleman conjecture implies $(\bar Y_1\times Y_2)\sslash T$ is mirror to $(\pi^\vee)^{-1}(1)=Y_1^\vee\times_{\check T_\C}Y_2^\vee$, which is the statement for 3d mirror symmetry in the hom level, namely
    \[
    _A\Hom_{T^*[\pt/T_\C]}([Y_1\sslash T],[Y_2\sslash T])\simeq\  _B\Hom_{T^*\check T_\C}(Y_1^\vee,Y_2^\vee).
    \]
\end{enumerate}
\end{rem}

\subsection{3d mirror symmetry for hypertoric varieties/stacks}\label{2.2}

In this subsection, we let $X$ be a hypertoric variety/stack, and $X^!$ be the Gale dual hypertoric stack (see below). The 3d mirror symmetry between $X$ and $X^!$ are described in the brane level by an SYZ-type transform.

Let $\mathbb{P}_\Sigma$ be a toric variety with the dense torus $T^n_\C$ and fan $\Sigma$. Let $\Sigma(1)=\{u_1,\dots,u_l\}$ be the set of rays in $\Sigma$. Let $T_\C$ be a complex torus, and $\rho: T_\C\to T^n_\C$ be a homomorphism. Let $C=[\mathbb{P}_\Sigma/T_\C]$.
\begin{df}\label{galedef}
    The Gale dual $C^!$ of $C$ is the toric stack
    \[[(\C^l\times \check T_\C)/\check T^n_\C].\]
    Here, $\check T^n_\C$ acts on $\C^l$ through $(u_1,\dots,u_l):\check T^n_\C\to (\Cx)^l$, and acts on $\check T_\C$ through the dual homomorphism $\check \rho:\check T^n_\C\to \check T_\C$. We may also write $\C^{\Sigma(1)}$ in place of $\C^l$ to emphasize the $\check T^n_\C$-action.
\end{df}
In this subsection, we let $X=T^*C$ and $X^!=T^*C^!$. In other words,  
\[X=[T^*\PP_\Sigma\sslash T_\C],   X^!=[T^*(\C^l\times \check T_\C)\sslash \check T^n_\C].\]

If $T_\C$ is trivial, and $\PP_\Sigma$ is the toric variety associated with the short exact sequence
\[1\to T^{l-n}_\C \to (\Cx)^l  \to T^{n}_\C \to 1\]
of complex tori, then $C^!$ is the toric stack associated with the dual short exact sequence
\[1\to \check T^{n}_\C \to (\Cx)^l  \to \check T^{l-n}_\C \to 1.\]

On the other hand, if $\PP_\Sigma=\pt$, then $X=T^*[\pt/T_\C]$ and $X^!=T^*\check T_\C$ is the 3d mirror pair we considered in the previous subsection.

\begin{rem}
    Note that even $X$ is a smooth complex symplectic manifold, $X^!$ may still be stacky. In fact, to obtain a non-stacky version of $X^!$, one has to choose a K\"ahler parameter of $X^!$, which corresponds to an equivariant parameter of $X$. See \Cref{matchingKahEqu}. 
\end{rem}

We will describe the SYZ type transform from $\Br_A(X;C)$ to $\Br_B(X^!;C^!)$. We first discuss a special case: the 3d mirror symmetry between the A-model of $X=T^*T^n_\C$ and B-model of $T^*[ \pt/\check T^n_\C] $, picturally
\begin{center}
 \begin{tabular}{ c c c }
\\
 \begin{tabular}{ c }
 A-model \\
 $T^*T^n_\C$
\end{tabular} 
& $\underleftrightarrow{\text{ 3d mirror symmetry }}$ & 
 \begin{tabular}{ c }
 B-model \\
 $T^*\check[\pt/\check T^n_\C]$.
\end{tabular} \\
\\
\end{tabular}
\end{center}
This is reversing the A- and B-sides of the 3d mirror symmetry we discussed
in the last subsection. 

A 3d brane in $\Br(T^*T^n_\C;T^n_\C)$ is a K\"ahler manifold $Y$, together with a holomorphic map $Y\to T^n_\C$. By reversing $X$ and $X^!$ from the previous subsection, one would expect the SYZ mirror $Y^\vee$ of $Y$ should admit a Hamiltonian $\check T^n$-action, providing a 3d brane in $\Br(T^*[\pt/T^n_\C];[\pt/T^n_\C])$.

Recall that $Y^\vee$ is constructed as a moduli space of Lagrangians $L \subset Y$ and unitary flat line bundles $E$ on $L$, and $\check T^n$ parametrizes unitary flat line bundles over $T^n_\C$. The $\check T^n$-action on $Y^\vee$ can be realized by pulling back the unitary flat line bundles on $T_\C$ to $L$ and tensoring them with $E$. As we will demonstrate in \Cref{Aurouxpicture}, this action is both holomorphic and Hamiltonian (\Cref{G-hamaction}). In other words, the Hamltonian $\check T^n$-action on $Y^\vee$ provides a 3d brane in $\Br(T^*[\pt/\check T^n_\C];[\pt/T^n_\C])$.

\begin{exmp}\label{exmpCxm}
    \item The 2d mirror of $\left( \mathbb{C}^{\times }\right)^{m}$ is $%
\left( \mathbb{C}^{\times }\right) ^{m}$ with dual coordinates $z_{j}$'s and 
$z^{j}$'s respectively. With the diagonal action $\rho $ on $(\Cx)^m$ by $\mathbb{C}%
^{\times }$, we have 
\begin{equation*}
\begin{tikzcd}
    (\Cx)^m\ar[d,"\prod_{j=1}^{m}z_{j}"]\\
    \Cx
\end{tikzcd}
\text{ is 3d mirror to }
\begin{tikzcd}
    (\Cx)^m \ar[r,symbol=\curvearrowleft,"\rho"]&\Cx.
\end{tikzcd}
\end{equation*}%
It is convenient to rephrase ``$(\Cx)^m\curvearrowleft \Cx$'' as ``$[(\Cx)^m/\Cx]\to [\pt/\Cx]$''.
\end{exmp}

Before discussing the general case, we provide an example first.
\begin{exmp}\label{exmpcm}
In this example, we discuss the 3d mirror symmetry between the A-model of $X=T^*\mathbb{C}$ and B-model of $T^*\left[ \mathbb{C}/\mathbb{C}^{\times }\right]$, picturally
\begin{center}
 \begin{tabular}{ c c c }
\\
 \begin{tabular}{ c }
 A-model \\
 $T^*\C$
\end{tabular} 
& $\underleftrightarrow{\text{ 3d mirror symmetry }}$ & 
 \begin{tabular}{ c }
 B-model \\
 $T^*[\C/\Cx]$
\end{tabular} \\
\\
\end{tabular}
\end{center}
This 3d duality will be explained via the duality between partial compactification and potential function in mirror symmetry. 

For example, $\mathbb{C}^{m}$ is a partial compactification of $\left(\mathbb{C}^{\times }\right) ^{m}$ by adding union of coordinate hyperplanes $\bigcup _{j=1}^{m}H_{j}$ where $H_{j}=\left\{ x_{j}=0\right\}$; its mirror is the LG model $W=\sum_{j=1}^{m}x^{j}:\left( \mathbb{C}^{\times }\right)^{m}\rightarrow \mathbb{C}$, where each term $x^{j}$ in $W$ is the contribution of Maslov index 2 holomorphic disks meeting $H_{j}$. We write
\begin{equation*}
    \begin{tikzcd}
         \C^m\ar[d,"\prod x_j"]\\ \C
    \end{tikzcd}
    \text{ is 3d mirror to }
        \begin{tikzcd}
         (\Cx)^m\ar[r,symbol=\curvearrowleft]\ar[d,"\sum x^j"]&\Cx\ar[d,equal]\\
         \C\ar[r,symbol=\curvearrowleft]&\Cx
    \end{tikzcd}
\end{equation*}
where $\Cx$ acts on $\C$ via the standard action, and acts on $(\Cx)^m$ via the diagonal action.
This is an example of a transform from $\Br_A(T^*\C;\C)$ to $\Br_B(T^*[\C/\Cx])$.

Note that the A-sides of this example and \Cref{exmpCxm} are compatible, as seen by the diagram
\begin{equation*}
\begin{tikzcd}
    (\Cx)^m\ar[r,symbol=\subset]\ar[d,"\prod x_j"]&\C^m\ar[d,"\prod x_j"]\\
    \Cx\ar[r,symbol=\subset]&\C.
\end{tikzcd}
\end{equation*}
On the mirror side, it corresponds to the diagram

\begin{equation*}
\begin{tikzcd}
    \left[(\Cx)^m/\Cx\right]\ar[r,equal] \ar[d]&\left[(\Cx)^m/\Cx\right]\ar[d,"\sum x^j"]\\
    \left[\pt/\Cx\right]&\left[\C/\Cx\right] \ar[l].
\end{tikzcd}
\end{equation*}
The readers may refer to \Cref{SYZtype} for functoriality of 3d mirror symmetry of this type.
\end{exmp}

We now discuss the general case: 3d mirror symmetry between the A-model of $X=T^*C$ and B-model of $X^!=T^*C^!$ for a pair of Gale dual toric stacks $C$ and $C^!$.
\begin{center}
 \begin{tabular}{ c c c }
\\
 \begin{tabular}{ c }
 A-model \\
 $X=T^*[\PP_\Sigma/T_\C]$
\end{tabular} 
& $\underleftrightarrow{\text{ 3d mirror symmetry }}$ & 
 \begin{tabular}{ c }
 B-model \\
 $X^!=T^*[\C^l\times \check T_\C/\check T^n_\C]$
\end{tabular} \\
\\
\end{tabular}
\end{center}

Recall that $C=[\PP_\Sigma/T_\C]$ and $C^!=[(\C^l\times \check T_\C)/\check T^n_\C]$. Let $Y$ be a K\"ahler manifold, and $\pi:Y\to \PP_\Sigma$ be a $T$-equivariant holomorphic map. In other words, $[Y\sslash T]$ is a 3d brane in $\Br_A(X;C)$.

Let $Y^\circ=\pi^{-1}(T^n_\C)$, we have the pullback diagram
\begin{equation*}
    \begin{tikzcd}
        Y^\circ\ar[r,symbol=\subset]\ar[d,"\pi|_{Y^\circ}"]& Y\ar[d," \pi"]\\
        T^n_\C\ar[r,symbol=\subset]&\PP_\Sigma.
    \end{tikzcd}
\end{equation*}
If $Y^\vee$ is an SYZ mirror of $Y^\circ$, then there is a Hamiltonian $\check T^n$-action on $Y^\vee$ that is mirror to the map $Y^\circ\to T^n_\C$. On the other hand, there is a holomorphic map $\pi^\vee:Y^\vee\to \check T_\C$ that is mirror to the $T$-action on $Y^\circ$.

Recall $Y^\vee$ is regarded as the moduli space of Lagrangians with unitary flat line bundles $(L,E)$. Let $\beta\in H_2(L,Y)$ be a disc class, then $z_{\beta}(L,E)=e^{-\int_{\beta} \omega} \cdot \hol_{\partial \beta}(E)$ defines a (local) holomorphic function on $Y^\vee$ (see \Cref{Aurouxpicture}). 

For $j=1,\dots,l$, let $H_j$ be the toric divisor of $\PP_\Sigma$ corresponding to $u_j$, $D_j=\pi^{-1}(H_j)$, and $W_j:Y^\vee\to \C$ be the function
\[W_j(L,E)=\sum_{\beta\cap D_i=\delta_{ij}}n_\beta z_\beta,\]
where $n_\beta$ is the number of holomorphic discs in the class $n_\beta$. 

In \Cref{Aurouxpicture}, we will show that the map
\[F=(F_1,\dots,F_l,\pi^\vee):Y^\vee\to \C^l\times \check T_\C\]
is $\check T^n$-equivariant. Therefore, $[Y^\vee\sslash \check T^n]$ is a 3d brane of $X^!$ supported on $C^!$\footnote{The idea of using multipotential for LG models are known to the expects, see, for example, the last section in \cite{AurouxFibration} and \cite{Sukjoo}}.

\begin{exmp}
    Suppose $\Cx$ acts only on the first coordinate of $\C^n$, then 
\begin{equation*}
\begin{tikzcd}
    \left[\C^n/\Cx\right]\ar[d,"\prod_{j=1}^{n}x_{j}"]\\
    \left[\C/\Cx\right]
\end{tikzcd}
\text{ is 3d mirror to }
\begin{tikzcd}
    \left[(\Cx)^n/\Cx\right]\ar[r,symbol=\cong]\ar[d,"\sum_{j=1}^{n}x_{j}"]&(\Cx)^{n-1}\ar[d,"1+\sum_{j=2}^{n}x_{j}"]\\
    \left[\C\times \Cx/\Cx\right]\ar[r,symbol=\cong]&\C.
\end{tikzcd}
\end{equation*}
This is an example explaining the 3d mirror symmetry of the A-model of $T^*[\C/\Cx]$ and the B-model of $T^*\C$.
\begin{center}
 \begin{tabular}{ c c c }\\
 \begin{tabular}{ c }
 A-model \\
 $T^*[\C/\Cx]$
\end{tabular} 
& $\underleftrightarrow{\text{ 3d mirror symmetry }}$ & 
 \begin{tabular}{ c }
 B-model \\
 $T^*\C$.
\end{tabular} \\
\\
\end{tabular}
\end{center}
\end{exmp}

\subsection{Examples of functoriality of 3d mirror symmetry}\label{stratified}
It is expected that complex Lagrangian correspondences should induce transforms of 3d branes. For example, the diagonal 
$\Delta_{T^*\Cx}\subset T^*\C\times T^*\Cx$
induces a transform
\[\Phi_1:\Br_A(T^*\C;\C)\to \Br_A(T^*\Cx;\Cx).\]
Explicitly, $\Phi_1$ sends a 3d brane $\pi:Y\to \C$ of $T^*\C$ to the 3d brane $\pi:Y^\circ \to \Cx$, where $Y^\circ=\pi^{-1}(\Cx)$ as in the previous section.

On the other hand, the projection $[\C/\Cx]\to [\pt/\Cx]$ induces a complex Lagrangian in $T^*[\C/\Cx]\times T^*[\pt/\Cx]$, and hence gives a transform
\[\Phi_1^!:\Br_B(T^*[\C/\Cx];[\C/\Cx])\to \Br_B(T^*[\pt/\Cx];[\pt/\Cx]).\]
Explicitly, $\Phi_1^!$ sends a 3d brane $[Z\sslash S^1]\to [\C/\Cx]$ to the 3d brane $[Z\sslash S^1]\to [\pt/\Cx]$. \Cref{exmpcm} shows that $\Phi_1$ and $\Phi_1^!$ corresponds to each other under 3d mirror symmetry. In other words, the following commutative diagram commutes.
\[
\begin{tikzcd}
    \Br_{A}(T^*\C;\C)\ar[d,"\Phi_1"]\ar[rrr,leftrightarrow,"\text{ 3d MS}"]&&&\Br_{B}(T^*[\C/\Cx];[\C/\Cx]),\ar[d,"\Phi_1^!"]\\
    \Br_{A}(T^*\Cx;\Cx)\ar[rrr,leftrightarrow,"\text{ 3d MS}"]&&&\Br_{B}(T^*[\pt/\Cx];[\pt/\Cx]).
\end{tikzcd}
\]

Similarly, one may expect the following. The complex Lagrangian $\C\times \pt\subset T^*\C\times \pt$ induces a transform 
\[
\Phi_2:\Br_A(T^*\C,\C)\to \Br_A(\pt);
\]
and the complex Lagrangian $\Delta_{[\Cx/\Cx]}\subset T^*[\C/\Cx]\times \pt= T^*[\C/\Cx]\times T^*[\Cx/\Cx]$ induces a transform 
\[
\Phi_2^!:\Br_B(T^*[\C/\Cx],[\C/\Cx])\to \Br_B(\pt).
\]
We expect $\Phi_2$ and $\Phi_2^!$ corresponds to each other under 3d mirror symmetry\footnote{Functoriality similar to the following is also discussed in \cite{GHM}.}.
\begin{equation}\label{comdiag2}
    \begin{tikzcd}
    \Br_{A}(T^*\C;\C)\ar[d,"\Phi_2"]\ar[rrr,leftrightarrow,"\text{ 3d MS}"]&&&\Br_{B}(T^*[\C/\Cx];[\C/\Cx])\ar[d,"\Phi_2^!"]\\
    \Br_{A}(\pt)\ar[rrr,leftrightarrow,"\text{ 3d MS}"]&&&\Br_{B}(\pt).
\end{tikzcd}
\end{equation}
Explicitly, $\Phi_2$ sends a 3d brane $\pi:Y\to \C$ of $T^*\C$ to the (possibly singular) 3d brane $\pi^{-1}(0)$, and $\Phi_2^!$ sends a 3d brane $F:[Z\sslash S^1]\to [\C/\Cx]$ to the 3d brane $F^{-1}(\Cx/\Cx)$.

\begin{exmp}\label{exmpAAK}
    As in \Cref{exmpcm},
    \begin{equation*}
    \begin{tikzcd}
         \C^m\ar[d,"\pi=\prod x_j"]\\ \C
    \end{tikzcd}
    \text{  is 3d mirror to }
        \begin{tikzcd}
         \left[(\Cx)^m/\Cx\right]\ar[d,"F=\sum x^j"]\\
         \left[\C/\Cx\right].
    \end{tikzcd}
\end{equation*}
The fibres of $\pi$ over the stractum $\Cx\subset \C$ is isomorphic to $(\Cx)^{m-1}$, which is mirror to $F^{-1}([0/\Cx])\cong (\Cx)^{m-1}$.

In contrast, the zero fiber of $\pi$ is given by $\pi^{-1}(0)={\prod x_j=0}\subset \C^m$. Meanwhile, the free quotient $\{\sum x^j\neq 0\}/\Cx$ is isomorphic to the $(m-1)$-dimensional pair of pants. These two spaces indeed constitute a 2d mirror pair; see, for instance, \cite{Lekilipairofpants}.

\end{exmp}

In general, let $\PP_{\Sigma_1}$ be a toric subvariety of $\PP_\Sigma$ with fan $\Sigma_1$. We define $C_1=[\PP_{\Sigma_1}/T_\C]\subset [\PP_{\Sigma}/T_\C]=C$. We expect there is a commutative diagram

\begin{equation}\label{LagCorrespondence}
    \begin{tikzcd}
    \Br_{A}(T^*C;C)\ar[d,"\Phi"]\ar[rrr,leftrightarrow,"\text{ 3d MS}"]&&&\Br_{B}(T^*C^!;C^!)\ar[d,"\Phi^!"]\\
    \Br_{A}(T^*C_1;C_1)\ar[rrr,leftrightarrow,"\text{ 3d MS}"]&&&\Br_{B}(T^*C_1^!;C_1^!).
\end{tikzcd}
\end{equation}
Here, $\Phi$ is the restriction of a 3d brane from $C$ to the substack $C_1$. 

On the other hand, $\Phi^!$ sends a 3d brane $F:[Z\sslash \check T^n]\to C^!$ to (i) the composition $[Z\sslash \check T^n]\to C^!\to C_1^!$ if $\PP_{\Sigma_1}\subset \PP_{\Sigma}$ is an open toric subvariety; (ii) the restriction $F^{-1}(C_1^!)\to C^!_1$ if $\PP_{\Sigma_1}\subset \PP_{\Sigma}$ is a closed toric subvariety. The general case is a combination of (i) and (ii).

\subsection{SYZ type construction of 3d mirror manifolds}\label{SYZtype}

Recall that SYZ proposal says that mirror Calabi-Yau manifolds $Y$ and $Y^\vee$ admit dual Lagrangian torus fibrations. Furthermore mirror duality should
be realized as family of Fourier type transformation between Lagrangian tori 
$T$ in $Y$ and the corresponding Lagrangian tori $\check T$ in $Y^{\vee}$. Put
these differently, mirror symmetry comes from duality of 1d TQFTs along Lagranian tori $T$ in $Y$. We are going to explain how 3d mirror symmetry comes from mirror symmetry along complex Lagrangian tori $T_{\mathbb{C}}$ in $X$, at least in the Abelian case. 

Suppose $T_{\mathbb{C}}\subset X$ is a complex Lagrangian torus in $X$, a 3d
brane in $X$ with support on $T_{\mathbb{C}}$ is%
\[
\begin{array}{ccc}
Y &  &  \\ 
\downarrow  &  &  \\ 
T_{\mathbb{C}} & \subset  & X
\end{array}
\]
with $Y$ a K\"ahler manifold with holomorphic fibration $Y\rightarrow T_{\mathbb{C}}$. As explained earlier, the mirror of $Y\rightarrow T_{\mathbb{C}}$ is the mirror $Y^\vee$ of $Y$ together with a Hamiltonian $\check T$-action. Namely, $[ Y^{\vee}\sslash \check T] \rightarrow [\pt/\check T_\C]$ is a 3d brane in $T^*[\pt/\check T_\C]$. Thus we should have a complex Lagrangian correspondence in $T^*[\pt/\check T_\C] \times
X^{!}$ sending $[ Y^\vee/\check T_\C ] \to [\pt/\check T_\C]$ to a 3d brane in $X^{!}$. This complex
Lagrangian correspondence is 3d mirror to an inclusion $T^*T_{\mathbb{C}}\subset X$.

Next, for simplicity we assume $T_{\mathbb{C}}=\mathbb{C}^{\times }$ at the
time being, suppose $\mathbb{C}\subset X$ is a complex Lagrangian in $X
$ and   
\[
\begin{array}{ccc}
Y &  &  \\ 
\downarrow  &  &  \\ 
\mathbb{C} & \subset  & X%
\end{array}%
\]%
is a 3d brane in $X$. The mirror to $Y\rightarrow \mathbb{C}$ is given by $[ Y^\vee/\mathbb{C}^{\times }] \rightarrow [ \mathbb{C}/\mathbb{C}^{\times }] $ which is a 3d brane in $T^{\ast }[\mathbb{C}/\mathbb{C}^{\times }] $. As before we should have a complex Lagrangian correspondence in $T^{\ast }[\mathbb{C}/\mathbb{C}
^{\times }] \times X^{!}$ sending $[ Y^\vee/\mathbb{C}^{\times}] \rightarrow [\pt/T_{\mathbb{C}}^{\vee}] $ to a 3d brane in $X^{!}$. 

Similarly, given a 3d brane  
\[
\begin{array}{ccc}
Y &  &  \\ 
\downarrow  &  &  \\ 
\PP^{1} & \subset  & X,
\end{array}%
\]%
the mirror to $Y\rightarrow \PP^{1}$ is a 3d brane $\left[ Y^\vee/\mathbb{C}^{\times }\right] \rightarrow\left[\mathbb{C}_{1,-1}^{2}/\mathbb{C}^{\times }\right] $ in $T^{\ast }\left[ \mathbb{C}_{1,-1}^{2}/\mathbb{C}^{\times }\right]$. Here the subscript $1,-1$ indicate the weights of the $\mathbb{C}^{\times }$-action. 

Note that there is a pushout diagram%
\[
\begin{array}{ccc}
\mathbb{C}^{\times } & \rightarrow  & \mathbb{C} \\ 
\downarrow  &  & \downarrow  \\ 
\mathbb{C} & \rightarrow  & \PP^{1}%
\end{array}%
\]%
which corresponds to the pullback diagram 
\[
\begin{array}{ccc}
\left[ \pt/\mathbb{C}^{\times }\right]  & \leftarrow  & \left[ \mathbb{C}/%
\mathbb{C}^{\times }\right]  \\ 
\uparrow  &  & \uparrow  \\ 
\left[ \mathbb{C}/\mathbb{C}^{\times }\right]  & \leftarrow  & \left[ 
\mathbb{C}_{1,1}^{2}/\mathbb{C}^{\times }\right] 
\end{array}
\]
in the 3d mirror side. They give the various complex Lagrangian correspondences relating various 3d
mirror pairs. 

In \Cref{Gluing}, we use this idea to construct the Coulomb branch of a gauge theory with matters by gluing copies of the Coulomb branch of the pure gauge theory (see also \Cref{Gluingintro}).

\subsection{Related works}\label{relatedworks}{}

\begin{exmp}\label{LogCY}
    $T^*[\C/\Cx] \xleftrightarrow{\text{ 3d MS }} T^*\C$ and mirror symmetry for a log Calabi-Yau pair $(Y,D)$.

Let $Y$ be a K\"{a}hler manifold, and $D=\left\{ s=0\right\} \subset Y$ be
an anti-canonical divisor. Let $K_Y$ be the canonical bundle of $Y$, and $K_{Y}^{\times }$ be the corresponding $\mathbb{C}^{\times }$-bundle on $Y$. The section $s$ gives a $\Cx$-equivariant map $\phi_s:K_{Y}^{\times}\rightarrow \mathbb{C}$,
or equivalently, a 3d brane $Y=[K_{Y}^{\times }/\Cx] \to  [\C/\Cx]$ in $\Br_A(T^*[\C/\Cx];[\C/\Cx])$. We have

\begin{equation*}
    \begin{tikzcd}
         Y\ar[r,symbol=\cong]&K^\times_Y/\Cx\ar[d]\\ &\left[\C/\Cx\right]
    \end{tikzcd}
    \text{ is 3d mirror to }
        \begin{tikzcd}
         (K_Y^\times)^\vee/\Cx\ar[r,symbol=\cong]\ar[d,"\widetilde{f}/\Cx"]&Y^\vee\ar[d,"f"]\\\left[\C\times \Cx\ar[r,symbol=\cong]/\Cx\right]&\C.
    \end{tikzcd}
\end{equation*}
Here $(K_Y^\times)^\vee$ is the SYZ mirror of $\{\phi_s\neq 0\}\subset K_Y$, and $\Tilde{f}$ is obtained by counting Maslov two discs intersecting ${\phi_s=0}$; in other words, $((K_Y^\times)^\vee,\widetilde{f})$ is an SYZ mirror of $K_Y^\times$. As mentioned earlier, the conjecture of Teleman asserts that $(Y^\vee,f)$ is a 2d mirror of $Y$, see \cite{LLL} for a proof of the (modified) Teleman conjecture.

Now, if we apply the commutative diagram (\ref{comdiag2}) with $C_1=[\pt/\Cx]$, then we obtain
\begin{equation*}
\begin{tikzcd}
D\ar[d]\\ 
\left[\pt/\Cx\right]
\end{tikzcd}
\text{ is 3d mirror to }
\begin{tikzcd}
    f^{-1}(\Cx)\ar[d]\\
    \Cx.
\end{tikzcd}
\end{equation*}
In particular, we expect that $D$ is 2d mirror to a (generic) nonzero fibre of $f$. This expectation, and that $Y^\vee$ is 2d mirror to $Y\setminus D$ are well-known principles of relative mirror symmetry for log Calabi-Yau pairs.
\end{exmp} 

More generally, suppose we have a decomposition $D=D_1\dots+D_n$, then this provide a 3d brane $Y\to [\C^n/(\Cx)^n]$ of $T^*[\C^n/(\Cx)^n]$. Let $L_i$ be the line bundle corresponding to $\mathcal{O}_X(D)$, and $L_i^\times$ be the associated $\Cx$-bundle on $X$. If we denote 
\[
P=L^\times_1\times_YL^\times\times_Y\cdots \times_Y L^\times_n,
\]
then we expect 
\begin{equation*}
    \begin{tikzcd}
         Y\ar[r,symbol=\cong]&P/(\Cx)^n\ar[d]\\ &\left[\C^n/(\Cx)^n\right]
    \end{tikzcd}
    \text{ is 3d mirror to }
        \begin{tikzcd}
         P^\vee/(\Cx)^n\ar[r,symbol=\cong]\ar[d,"\widetilde{f}/\Cx"]&Y^\vee\ar[d,"f"]\\\left[\C^n\times (\Cx)^n\ar[r,symbol=\cong]/(\Cx)^n\right]&\C^n.
    \end{tikzcd}
\end{equation*}
In particular, we expect $f^{-1}(p)$ is mirror to $\bigcap_{i=1}^n D_i$ for a generic $p\in \C^n$ (apply (\ref{comdiag2}) with $C_1=[\pt/(\Cx)^n]$).

It would be interesting to see what happens when we impose linear equivalences among $D_1,D_2,\dots,D_n$, as illustrated by the next example.

\begin{exmp} \label{DHT}
$T^*[\C^2_{1,-1}/\Cx] \xleftrightarrow{ 3d MS }T^*[\C^2_{1,1}/\Cx]$ and Doran-Harder-Thompson conjecture.

Let $f:\mathcal{V}\rightarrow \mathbb{C}$ be a flat morphism (and $\mathcal{V}$ is smooth) such that $f^{-1}(0)=Y_{1}\cup _{Y_{12}}Y_{2}$ a simple normal crossing divisor with
two components of quasi-Fano varieties, and all other fibres are smooth
Calabi-Yau varieties. In particular $Y_{12}$ is an anti-canonical divisor in
both $Y_{1}$ and $Y_{2}$; moreover $O_{\mathcal{V}}\left( Y_{1}\right)
\simeq O_{\mathcal{V}}\left( -Y_{2}\right) $ which we denote as $L$. Then 
$Y_{1}$ and $Y_{2}$ induces a $\mathbb{C}^{\times }$-equivariant map $%
L^{\times }\rightarrow \mathbb{C}_{1,-1}^{2}$, where $L^{\times }$ is the
associated $\mathbb{C}^{\times }$-bundle of $L$ over $\mathcal{V}$. Namely $
\mathcal{V}\cong L^{\times }/\mathbb{C}^{\times }\rightarrow \lbrack 
\mathbb{C}_{1,-1}^{2}/\mathbb{C}^{\times }]$ defines a 3d brane in $T^{\ast
}[\mathbb{C}_{1,-1}^{2}/\mathbb{C}^{\times }]$.

Let $(Y_{i}^{\vee },W_{i})$ be a 2d mirror of $(Y_{i},Y_{12})$ for $i=1,2$. \citeauthor{DHT} \cite{DHT} proposed that the 2d-mirror of a generic fiber $U$ of $f$ is obtained from gluing $W_1:Y_{1}^{\vee }\to \C$ and $W_2:Y_{2}^\vee\to\C$ to form a family $U^\vee\to \PP^1$. 

In our framework, we have
\begin{equation*}
    \begin{tikzcd}
         \mathcal{V}\ar[r,symbol=\cong]&L^\times/\Cx\ar[d,"\pi"]\\ &\left[\C^2_{1,-1}/\Cx\right]
    \end{tikzcd}
    \text{ is 3d mirror to }
        \begin{tikzcd}
         \check {L}^\times/(\Cx)^2\ar[r,symbol=\cong]\ar[d]&U^\vee\ar[d,"F"]\\\left[\C^2_{1,1}\times (\Cx)^2/\Cx\right]\ar[r,symbol=\cong]&\left[\C^2_{1,1}/\Cx\right].
    \end{tikzcd}
\end{equation*}
Note that $\mathbb{P}^{1}$ is a GIT quotient of $\mathbb{C}_{1,1}^{2}$ by $
\mathbb{C}^{\times }$. 

On the other hand, applying (\ref{comdiag2}) with $C_1=[0\times \C_1/\Cx]$ tells us that
\begin{equation*}
    \begin{tikzcd}
         &Y_1\ar[d,"\pi"]\\
         \left[0\times \C_1/\Cx\right]\ar[r,symbol=\cong]&\left[\C/\Cx\right]
    \end{tikzcd}
    \text{ is 3d mirror to }
        \begin{tikzcd}
         F^{-1}(\Cx\times \C/\Cx)\ar[d]&\\
         \C\ar[r,symbol=\cong]&\left[\Cx\times\C_1/\Cx\right].
    \end{tikzcd}
\end{equation*}
Combining with \Cref{LogCY}, we see that $(F^{-1}(\Cx\times \C/\Cx), F)$ should be a 2d mirror of $(Y_1,Y_{12})$ (similar for $(Y_2, Y_{12})$). This agrees with the construction in the DHT proposal.
\end{exmp}

\begin{exmp}\label{AAk}
    $T^*\C\xleftrightarrow{ 3d MS }T^*[\Cx/\Cx]$ and the Abouzaid-Auroux-Katzarkov mirror symmetry for hypersurface in $(\Cx)^n$.
    
Let $Y$ be a toric Calabi-Yau of dimension $d$, and $\pi:Y\to \C$ the function with $\pi^{-1}(0)$ equal to the union of all toric divisors. Then
\begin{equation*}
    \begin{tikzcd}
         Y\ar[d,"\pi"]\\ \C
    \end{tikzcd}
    \text{  is 3d mirror to }
        \begin{tikzcd}
         \left[(\Cx)^d/\Cx\right]\ar[d,"F=\sum_{i=1}^m z^{u_i}"]\\
         \left[\C/\Cx\right],
    \end{tikzcd}
\end{equation*}
where $u_1,\dots,u_m$ are the ray generators of the fan associated to $Y$. According to \Cref{stratified}, we expect $\pi^{-1}(0)$ is 2d mirror to $F^{-1}([\Cx/\Cx])$ (apply (\ref{comdiag2}) with $C_1=0$).

Note that $F^{-1}([\Cx/\Cx])$ is isomorphic to the hypersurface $Z$ in $(\Cx)^m$ defined by the equation
\[\sum_{i=1}^m z^{u_i}=1.\]

On the other hand, by Kn\"orrer periodicity, we have (\cite{Orlov}, see \cite{Jeffs} for the corresponding statement in A-side):
\[D^b(\pi^{-1}(0))\simeq MF(Y\times \C_x,x\pi).\]
One can show that $(Y\times \C_x,x\pi)$ is the mirror of $Z$ proposed in \cite{AAK} (see also \cite{AAhypersurface}). Note that \Cref{exmpAAK} is a special case.
\end{exmp}

\begin{exmp}
    If $Y=C$, and $\pi$ is the identity, then $[Y^\vee/\check T^n_\C]=[\check T^n_\C/\check T^n_\C]$ is a point. The cotangent fibre $L$ of this point in $C^!$ is a section of the complex moment map (cocore). This agrees with the considerations of category $\mathcal{O}$. Where $C$ is the support of a simple object in the category $\mathcal{O}$ for $X$, and the conormal $L$ is the support of a projective object in the category $\mathcal{O}$ for $X^!$\footnote{We thank Justin Hilburn for discussions related to this}, see \cite{BLPWhypertoric,CGH,Hilburnthesis}.
\end{exmp}

\subsection{Non-Abelian equivariant mirror symmetry}\label{nonabelian}
In this subsection, we will consider the case when $X=T^*[\pt/G_\C]$, where $G_\C$ is the complexification of a compact group $G$.

In this case, the 3d mirror $X^!$ is not necessarily a cotangent bundle, so we cannot stay in the zero section level. There is a way to move an LG family over the zero section to a complex family over another Lagrangian, as we are describing now.

We begin with a more general situation. Suppose $\pi: Y\to B$ is a holomorphic submersion of complex manifolds, and $W:Y\to \C$ is a holomorphic function. Let $p\in Y$ and $F=\pi^{-1}(\pi(p))$ be the fibre of $\pi$ containing $p$. Note that we have an exact sequence
\[0\to T^*_{\pi(p)}B\to T^*_pY\to T^*_{p}F\to 0,\]
and $dW$ is an element of $T^*_pY$. If $dW(p)$ vanishes along $F$, then the above sequence says that we can regard $dW$ as an element of $T^*_{\pi(p)}B$. We write $\widetilde C=\widetilde C_Y\subset Y^\vee$ for the fibrewise critical locus of $W$, and define
\[
C=C_Y=\{dW(p)\in T^*_{\pi(p)}B:dW \in \widetilde C_Y\}\subset T^*B.
\]
One can see that $C$ is a complex Lagrangian in $T^*B$. There is another description for $C$. The graph of $\pi$ is a submanifold of $Y\times B$, and its conormal defines a Lagrangian correspondence from $T^*Y$ to $T^*B$, which sends the graph of $dW$ to $C$. Note that this latter description is more general, and $C$ may be defined even $\pi$ is not a submersion.

\begin{rem}
    If we use the languages of shifted symplectic structures \cite{ShiftSym}, and define $\widetilde C$ as the derived critical locus, then $\widetilde{C}$ is also a Lagrangian, so we can work with $\widetilde{C}$ directly. Moreover, complex Lagrangian correspondence in this setting is always well-defined \Cite{LagCor}.
\end{rem}

Let $Y$ be a $T$-Hamiltonian K\"ahler manifold, then we can obtain an LG model $(Y^\vee,W)$, and a holomorphic map $W:Y^\vee\to \check T_\C$. We can apply the construction above to $W$ and $\pi$, and obtain a complex Lagrangian $C_Y$ in $T^*\check T_\C$.

\begin{exmp}\label{exam1}
 	Let $Y=\mathbb{P}^1$ with the standard toric $T=S^1$-action. The Hori-Vafa mirror of $Y$ is $(\Cx, W_Y=w+1/w)$, and the fibration $\Cx\to \Cx=\check T_\C$ is the identity map. In this case, $C_Y$ is simply the graph of $dW_Y$. If $z$ and $h$ are the base and fiber coordinates of $T^*\check T_\C=T^*\Cx$, then $C_Y$ is defined by the equation
 	\[h=z-\frac{1}{z}.\]
 \end{exmp}
 
  \begin{exmp}\label{exam2}
 	We still let $Y=\mathbb{P}^1$, but assume $T=S^1$ is the double cover of $S^1$ in \Cref{exam1}. In other words, $\{\pm 1\}\subset S^1$ lies in the kernel of the action. We still have $(Y^\vee,W_Y)=(\Cx,w+1/w)$, but $Y^\vee\to \check T_\Cx$ is the square map $w\mapsto z=w^2$. In this case $\widetilde C$ is isomorphic to $\Cx$ (with coordinate $w$), and the map $C\to T^*\Cx$ is given by
 	\[z=w^2,h=\frac{1}{2}(w-w^{-1}).\]
 \end{exmp}

When $X=T^*[\pt/G_\C]$, $X^!$ is the BFM space introduced by \cite{BFM}. $X^!$ is a smooth affine symplectic variety whose coordinate ring is the subalgebra of $W_G$-invariant elements of the algebra $A_{\check G}$ described below. Let $\Phi=\Phi(\check G,\check T)$ be the set of roots for $(\check G,\check T)$. Each $\alpha\in \Phi$ is considered as a function on $\check T_\C$, we write $h_{\alpha}$ for the corresponding element in $\lt_\C$, and $s_\alpha$ for the corresponding simple reflection. We let $h^\vee_{\alpha}\in \lt^*$ be such that $s_\alpha(h)=h- h^\vee_\alpha(h) h_\alpha$ for $h\in \lt_\C$

\begin{df}\label{BFM space}
    \begin{enumerate}\item[] 
    \item The affine blowup algebra $A^\circ_{\check G}$ for $\check G$ is \[A^\circ_{\check G}=\C\left[T^*\check T_\C,\left\{\frac{\alpha-1}{h_{\alpha}^\vee}\right\}_{\alpha\in \Phi}\right].\]
    \item $A_{\check G}$ is defined to be the intersection 
    \[A_{\check G}=\bigcap_{\alpha\in \Phi}\left(A^\circ_{\check G}\left[\frac{1}{h_{\alpha'}}\right]_{\alpha'\in \Phi\setminus \{\alpha\}}\right).\]
\end{enumerate}
\end{df}

Let $Y$ be a $G$-Hamiltonian manifold, and recall that we define $Y^\vee$ as the moduli space of $T$-invariant Lagrangians branes. It is clear that $Y^\vee$ has a natural $W_G=N(T)/T$ action, so that $\pi:Y^\vee\to \check T_\C$ is $W_G$-equivariant and $W$ is $W_G$-invariant. This implies $C_Y\subset T^*\check T_\C$ is $W_G$-invariant. We consider the diagram
\begin{equation*}
    \begin{tikzcd}
        T^*\check T_\C\ar[rd]&&\ar[dl]BFM(\check G)\\
    &T^*\check T_\C/W_G
    \end{tikzcd}
\end{equation*}
We proposed that $C_Y/W_G\subset T^*\check T_\C/W_G$ can be lifted to a complex Lagrangian in $BFM(\check G)$.

\begin{exmp}
 	In \Cref{exam1}, the $S^1$-action on $Y$ naturally extends to a $G=\operatorname{PSU}(2)$ action.
  As we have already seen, $C'\cong C\subset T^*\Cx$ is defined by equation $h=z-1/z$. $C$ is clearly $W_G=\Z_2$-invariant, and we have an extended ring homomorphism
    \[A^\circ_{\check G}=\C\left[z^{\pm1},h,\frac{z^2-1}{h}\right]\to\frac{\C\left[z^{\pm1},h\right]}{(h-(z-z^{-1}))}\]
    sending $\frac{z^2-1}{h}$ to $z$, which induces the lift $C\to \spec A^\circ_{\check G}$.
 \end{exmp}
 
  \begin{exmp}
 	In \Cref{exam1}, the $S^1$-action on $Y$ naturally extends to a $G=\operatorname{SU}(2)$ action.
As we have already seen $\widetilde C$ is isomorphic to $\Cx$ (with coordinate $w$), and we have a the map $\widetilde C\to T^*\Cx$ is given by $z=w^2,h=\frac{1}{2}(w-w^{-1})$. This map is $W_G=\Z_2$-equivariant, and we also have an extended ring homomorphism
 	\[A^\circ_{\check G}=\C\left[z^{\pm 1},h,\frac{z-1}{h}\right]\to\C\left[w^{\pm1}\right]\]
 	sending $\frac{z-1}{h}$ to $2w$, which induces the lift $C\to \spec A^\circ_{\check G}$.
 \end{exmp}

In \cite{paper2}, we prove the following theorems, and show that they imply the existence of a lift of $C_Y$ to $BFM(\check G)$.

\begin{thm}\label{Floerthm1}
Let $Y$ be a $G$-Hamiltonian manifold, $L\subset Y$ be a compact, connected, $T$-relatively spin, Lagrangian submanifold, and $E$ be a unitary flat line bundle on $L$. If $HF^\bullet((L,E),(L,E))\neq 0$, then $\fz(L,E)\triangleq\exp(-\mu_T(L))\hol_T(E)\in T_\C^\vee$ lies in the center for the Langlands dual group $ \check G$.
\end{thm}

\begin{thm}\label{Floerthm2}
Let $Y$ and $(L,E)$ be as above. Suppose the minimal Maslov number of $L$ is at least $2$, and $H^\bullet(L)$ is generated in degree $1$. Let $\alpha\in\Phi$, if $dW\in H_1(L)$ lies in the image of $H_1(T)^{s_\alpha}$ under the natural map
\[H_1(T)\to H_1(L),\]
then $\fz(L,E)\in \ker \alpha$.
\end{thm}

\begin{rem}
Let $V$ be a $G$-representation. \citeauthor{BFN} (\cite{BFN}) gives a construction of the Coulomb branch $Coul(G,V)$, which is (the affinization of) the 3d mirror of $T^*[V/G_\C]$. In their construction, they defined a space $R_{G,V}$, which is a closed subvariety of an infinite dimensional vector bundles over the affine Grassmannian $\operatorname{Gr}_{G_\C}$ for $G_\C$. The coordinate ring of the Coulomb branch $Coul(G,V)$ is the $G$-equivariant Borel-Moore homology $H_{BM}^G(R_{G,V})$\footnote{More precisely, it should be the $G_\C[[t]]$-equivariant Borel-Moore homology} of $R_{G,V}$.

There is a commutative diagram
\[\begin{tikzcd}
    H_{BM}^T(R_{G,V})\ar[r]&H_{BM}^T(R_{G})\\
    H_{BM}^T(R_{T,V})\ar[r]\ar[u]&H_{BM}^T(R_{T})\ar[u].
\end{tikzcd}\]

Suppose $Z$ is a $G$-Hamiltonian K\"ahler manifold, $(Z^\vee,W_Z)$ is its mirror, and $\pi_Z^\vee:Z^\vee\to \check T_\C$ is the Teleman map.

If we let $Y=V\times Z$, and $G$ actions on $Y$ with the diagonal action, then the mirror of $Y$ is $(Y^\vee=Z^\vee\times (\Cx)^l,W_Y=W_Z+\sum_{i=1}^l x_i)$ with Teleman map $\pi_Y^\vee(z,x)=\pi_Z^\vee\cdot \check \rho(x)$. We have the following commutative diagram
\[\begin{tikzcd}
    (Z^\vee,W_Z)\ar[r,symbol=\cong]\ar[d]&(Y^\vee,W_Z)/(\Cx)^l\ar[d]\\
    \check T_\C\ar[r,symbol=\subset]& \C^l\times T_\C^\vee /(\Cx)^l.
\end{tikzcd}
\]
Therefore, the problem of non-Abelian lifting reduces to the case when $V=0$.
\end{rem}

\subsection{Gluing}\label{Gluingintro}
In \Cref{Gluing}, we describe how to obtain the 3d mirror of $X=[T^*V\sslash T_\C]$ by gluing copies of 3d mirror of $T^*\check T_\C$, namely $T^*\check T_\C$. We will see an example in this subsection. We let $T_\C=\Cx$, and $V=\C$ be the standard representation.

Let $Y$ be a $S^1$-Hamiltonian K\"ahler manifold, then its mirror $(Y^\vee,W)$ should posses a holomorphic map $\pi^\vee :Y^\vee\to \Cx$. As discussed in \Cref{nonabelian}, we can obtain a complex Lagrangian $C_Y\subset T^*\Cx\cong \Cx_z\times \C_h$. By definition, $(z_1,h_1)\in C_Y$ if and only if there exists $p\in Y^\vee$ such that the equality
\begin{equation}\label{C_Y}
    z_1=\pi^\vee,\ dW=h_1d\log \pi^\vee
\end{equation}
holds at $p$.

On the other hand, if we consider the diagonal action of $S^1$ on $Y\times V\cong Y\times \C$, the Teleman map would be $\pi^\vee\cdot x :Y^\vee\times \Cx_x$, and $(z_2,h_2)\in C_{Y\times V}$ if and only if there exists $p\in Y^\vee$ and $x\in \Cx$ such that the equality
\begin{equation*}
    z_2=x\pi^\vee,\ dW+dx=h_2(d\log x\pi^\vee)
\end{equation*}
holds at $(p,x)$. The above equation can be simplified as $h_2=x$ and
\begin{equation}\label{C_YxC}
    z_2=h_2\pi^\vee ,\ dW(p)=h_2d\log \pi^\vee.
\end{equation}
holds at $p$. Comparing \Cref{C_Y,C_YxC}, we can glue $C_{Y}$ and $C_{Y\times V}$ by identifying $(z_1,h_1)$ with $(z_2,h_2)$ when
\begin{equation}\label{BFNC2}
    h_1=h_2,z_2=h_1z_1.
\end{equation}
Note that \Cref{BFNC2} is independent of $Y$, so we can use the same formula to glue two copies of $T^*\Cx$, the space we obtain would be $\C^2-(0,0)$, which is a subset of $T^*\C^2$. In \Cref{Gluing}, we will use the more abstract language of Lagrangian correspondence to derive the gluing formula, but the essential idea is the same.

\section{SYZ construction for the mirror 3d branes}\label{Aurouxpicture}
	In this section, we explain the 3d analogue of the SYZ type transform of 3d branes
 \[
 \Br(T^*C;C)\to \Br(T^*C^!;C^!)
 \]
 where $C$ is a toric stack, and $C^!$ is the Gale dual stack defined in \Cref{galedef}. We begin by recalling some Lie theorectic notions.

    Let $T_\C\cong (\Cx)^r$ be a complex torus with the (unique) maximal compact subgroup $T$, $\lt_\C$ (resp. $\lt$) be the Lie algebra of $T_\C$ (resp. $T$), and $\Lambda=\Lambda_T=2\pi i\lt_\Z\cong \pi_1(T)$ be the kernel of the exponential map $\exp:\lt\to T$. A character (resp. cocharacter) of $T$ is considered as an element of $\Hom_\Z(\lt_\Z, \Z)=\vcentcolon\lt_\Z^*$ (resp. $\Hom_\Z(\Z,\lt_\Z)=\lt_\Z$).  We fix an integral basis $\lambda_1,\lambda_2\,\dots,\lambda_r$ of $\Lambda$, and let $\ell_1,\ell_2,\dots,\ell_r\in \lt_\C^*=\Hom_\R(\lt,\C)$ be such that $\ell_i(\lambda_j)=\delta_{ij}$.

    If $\ell\in \lt_\C^*$, we write $d\ell$ for the corresponding left-invariant complex-valued 1-form on $T_\C$. We fix the holomorphic volume form
	\[\Omega_{T_\C}=i^r\bigwedge_{j=1}^rd\ell_j.\]
	
	If $z_1,\dots,z_r$ are the standard complex coordinates of $(\Cx)^r$, then the formula above reads
    \[\Omega_{(\Cx)^r}=\frac{1}{(2\pi i)^{r}} \bigwedge_{j=1}^r\frac{dz_j}{z_j}.\]


 We use the notation $\check T_\C$ to denote the dual group $\Hom_\Z(\Lambda,\Cx)$ of $T_\C$ with maximal compact subgroup $ \check T=\Hom_\Z(\Lambda,\operatorname{U}(1))$. The Lie algebra $\check \lt_\C=\Hom_\Z(\Lambda,\C)$ of $\check T_\C$ is canonically identified with $\lt_\C^*$. Moreover, we have
\begin{align*}
    \Lambda_{\check T}&=\Hom_\Z(\lt_\Z,\Z)\\
    \check \lt_\Z &=\Hom_\Z(\Lambda_{T},\Z).
\end{align*}
On the other hand, the Lie algebra $\check \lt$ of $\check T$ and the dual of $\lt$ are defined as
\begin{align*}
     \check \lt&=\Hom_\Z(\Lambda,i\R).\\
    \lt^*&=\Hom_\R(\lt,\R)
\end{align*}
Each of them can be identified with a real subspace of $\lt_\C^*$ with $\check\lt=i\lt^*$ and $\lt_\C^*=\lt^*\oplus i\lt^*$.

 If a Lie group $G$ acts on a smooth manifold $M$ and $\xi$ is an element of the Lie algebra of $G$, we denote by $\xi^\sharp$ the fundamental vector field $p\mapsto \frac{d}{dt}\big|_{t=0}\exp(t\xi)\cdot p$ induced by $\xi$ on $M$.

	

We next describe a generalization of the SYZ construction \cite{SYZ} for Calabi-Yau manifolds, and \cite{AurouxComp} for log Calabi- Yau manifolds. 

 \subsection{Equivariant SYZ type transform}
	Let $Y$ be a K\"ahler manifold of complex dimension $d$, and $T$ acts on $Y$ preserving the K\"ahler structure and has a moment map $\mu: Y\to \lt^*$\footnote{Our convention for moment map is $\omega(\cdot,\xi^\sharp)=\langle d\mu,\xi\rangle$ for $\xi\in \lt$.}. In other words, we have a 3d brane $[Y\sslash T]\to [\pt/T_\C]$ in $T^*[\pt/T_\C]$.
 
 We denote $\omega$, $J$, and $g$ the $T$-invariant K\"ahler form, complex structure, and Riemannian metric of $Y$. We may use the notations $\omega_Y$, $J_Y$, etc. if we want to specify their dependences on $Y$. Suppose $\Omega$ is a $T$-invariant meromorphic volume form on $Y$ so that $D\equiv -\operatorname{div}(\Omega)\geq 0$.
	\begin{df}\label{dfSlag}
		A connected compact Lagrangian submanifold $L\subset Y\setminus D$ is special if $\operatorname{Im}(\Omega)|_{L}=0$.
	\end{df}
	
    Let $L\subset Y\setminus D$ be a special Lagrangian and $\psi=\psi_L:L\to \R_{>0}$ be a function so that $\frac{1}{\psi}\operatorname{Re}(\Omega_Y)|_L$ is equal to the Riemannian volume form $\operatorname{vol}_L$ of $L$.
	\begin{df}
		A 1-form $\alpha\in \Omega^1(L,\R)$ is $\psi$-harmonic if $d\alpha=d(\psi*\alpha)=0$.
	\end{df}

    \begin{lem}\label{G-inv}
		Let $L\subset Y\setminus D$ be a $T$-invariant special Lagrangian, and $\psi:L\to \R$ be the $T$-invariant function such that
		\[\operatorname{Re}(\Omega_Y)|_L=\psi\operatorname{vol}_L,\]
		then every $\psi$-harmonic 1-form on $L$ is $T$-invariant.
	\end{lem}
	\begin{proof}
		Let $g_\psi$ be the $T$-invariant metric $\psi^{\frac{2}{d-2}}g|_L$ on $L$. Then a 1-form is $\psi$-harmonic if and only if it is harmonic for the $T$-invariant metric $g_\psi$, and is thus $T$-invariant. 
	\end{proof}

 Let $Y^\vee$ be the set of $\uL=(L,E)$, where $L\subset Y\setminus D$ is a $T$-invariant special Lagrangian submanifold and $E$ is an isomorphism class of unitary flat line bundles on $L$. We denote $NL$ the normal bundle of $L$ in $Y$. $Y$ has a differentiable manifold structure and the tangent space at a point $\uL\in Y^\vee$ can be identified with the vector space of pairs $(u,\alpha)$, where $u\in C^\infty(NL)$ is such that $\iota_u\omega$ is $\psi_L$-harmonic and $\alpha\in H^1(L;\R)$ (\cite{AurouxComp})\footnote{A closed 1-form $A$ on $L$ defines a flat $\operatorname{U}(1)$ connection $E=d-iA$ on $L$, and its holonomy on a closed curve $\gamma: S^1\to L$ is $\exp(i\int_\gamma A)$.}.

	\begin{df}\label{Mpt}
    Suppose that for each $\uL\in Y^\vee$, the sum
    \begin{equation}\label{Wsum}
			\sum_{\substack{\beta\in H_2(Y,L)\\\beta\cdot D=1}}n_\beta(L)\exp (-\int_\beta\omega)\operatorname{hol}_{\partial \beta}(E)
		\end{equation}
  converges, where $n_\beta(L)$ is the number of holomorphic discs in the class $\beta$. We define $W:Y^\vee\to \C$ to be the function that assigns each $\uL\in Y^\vee$ the value given by (\ref{Wsum}).
	\end{df}

	\begin{rem}
        We ignore the various assumptions to make $n_\beta(L)$ well-defined and the issues of convergence and wall crossing of $W$ in this paper. The readers can refer to \cite{AurouxComp, FOOOI,FOOOII} for more technical details. In this subsection, we will only prove statements concerning individual terms in (\ref{Wsum}), and the statements for $W$ follow as corollaries once the sum (\ref{Wsum}) is well-defined.
	\end{rem}
 
    It is known that $Y^\vee$ is a K\"ahler manifold and (each term in) $W$ is holomorphic. Furthermore, $Y^\vee$ is Calabi-Yau if we restrict $L$ to be special Lagrangian tori. These facts  follow from \Cref{complexstr}, \Cref{holoW} and \Cref{CalabiYau} below, whose proof can be found in \cite{AurouxComp}.
	
	\begin{lem}\label{complexstr}
		$Y^\vee$ has an integrable complex structure $J^\vee$. At a tangent space $T_{\uL} Y^\vee$, it is defined as 
		\[
		J^\vee:(u,\alpha)\mapsto(a,[-\iota_u\omega])
		\]
		where $a \in C^\infty(NL)$ is the unique section such that $\iota_a\omega$ is $\psi$-harmonic and represents the cohomology class $\alpha$. 
	\end{lem}
	\begin{lem}\label{holoW}
		For each $\beta\in H_2(Y,L;\Z)$,
		\[z_\beta(\uL)=\exp (-\int_\beta\omega)\hol_{\partial \beta}(E)\]
		is a holomorphic function. In particular, if the sum (\ref{Wsum}) converges absolutely, then $W$ is holomorphic.
	\end{lem}
	
	\begin{df}
		Given $(u_1,\alpha_1), (u_2,\alpha_2)\in T_{\uL}Y^\vee$, we define
		\begin{equation}\label{symplectic}
			\omega^\vee((u_1,\alpha_1),(u_2,\alpha_2))=\int_L\alpha_2\wedge\iota_{u_1}\operatorname{Im}(\Omega)-\alpha_1\wedge\iota_{u_2}\operatorname{Im}(\Omega).
		\end{equation}
		
		Similarly, suppose $L$ is a torus, and $(u_1,\alpha_1),\dots,(u_d,\alpha_d)\in T_{\uL}Y^\vee$, one can define
		
		\begin{equation}\label{holo-vol}
			\Omega^\vee((u_1,\alpha_1),\dots,(u_n,\alpha_d))=\int_L((-\iota_{u_1}\omega+i\alpha_1)\wedge\cdots\wedge (-\iota_{u_d}\omega+i\alpha_d)).
		\end{equation}
	\end{df}
	\begin{lem}\label{CalabiYau}
		$\omega^\vee$ is a K\"ahler form on $Y^\vee$, compatible with $J^\vee$; and $\Omega^\vee$ is a holomorphic volume form on the open subset of $Y^\vee$ formed by special Lagrangian tori.
	\end{lem}
 Henceforth, we will restrict $L$ to be a torus, although most propositions not involving $\Omega^\vee$ holds more generally.

	Suppose $\uL\in Y^\vee$, and $p\in L$, then the map $T\ni t \mapsto tp\in L$ induces a group homomorphism $\Lambda\cong \pi_1(T)\to \pi_1(L)$. The holonomy of $E$ therefore induces an element 
	\begin{equation*}
		\hol_T(E) \in\Hom_\Z(\Lambda,\operatorname{U}(1))=\check T,
	\end{equation*}
	which is independent of $p$ because we assume $L$ is connected. Since $L$ is $T$-invariant, $\mu$ is constant on $L$, we call this common value $\mu(L)$. Recall that $\check \lt=i\lt^*$.
	
	\begin{prop}\label{holo-pro}
		The function 
		\[\pi^\vee:Y^\vee\to \check T_\C=\exp(\lt^*)\cdot\check T\]
		\begin{equation}\label{formulapi}
			\uL\mapsto \exp(-\mu(L))\cdot \hol_T(E)
		\end{equation}
		is holomorphic.
	\end{prop}
	\begin{proof}
		Let $(u,\alpha)\in T_{\uL}Y$, and $\xi\in \lt$. We have
		\begin{equation}\label{critical}
			d\log\pi^\vee(u,\alpha)(\xi)=-\omega(u,\xi^\sharp)+i\alpha(\xi^\sharp)
			=(-\iota_u\omega+i\alpha)(\xi^\sharp),
		\end{equation}
		which is $\C$-linear in $(u,\alpha)$.
	\end{proof}

    The discussions below are related to the Teleman's conjecture concerning the mirror symmetry of symplectic reductions. Readers that are more interested in the SYZ construction of mirror 3d branes can jump to \Cref{fibration over complex tori}.
    
	As a consequence of (\ref{critical}), we have the following:
	\begin{cor}\label{criticaliff}
		Let $\uL=(L,E)\in Y^\vee$. Then $L$ contains a critical point of $\mu$, if and only if $\uL$ is a critical point of $\pi^\vee$.
	\end{cor}
	\begin{proof}
		From \Cref{critical}, one sees that $\uL$ is a critical point of $\pi^\vee$ if and only if the pullback $H^1(L,\R)\to H^1(T,\R)$ is not surjective.
		
		If $p\in L$ is a critical point of $\mu$, then $\xi^\sharp$ vanishes along the orbit $T\cdot p$ for some $\xi\in \lt$, so the pullback $H^1(L,\R)\to H^1(T\cdot p,\R)$ is not surjective.
		
		For the converse, note that a $\psi_L$-harmonic 1-form on $L$ is $T$-invariant by Lemma \ref{G-inv}. If $L$ does not contain a critical point of $\mu$, then the action of $T$ on $L$ is locally free. Since we assume $L$ is a torus, \Cref{lemmaTactsT} below implies the restriction map $H^1(L;\R)\to H^1(T;\R)$ is surjective.
	\end{proof}
	\begin{lem}\label{lemmaTactsT}
		Let $M$ be a compact, connected, orientable d-dimensional manifold such that the natural map
		\[\bigwedge^dH^1(M;\R)\to H^d(M;\R)\]
		is surjective. Suppose $T$ acts on $M$ in a locally free manner, and $p\in M$, then the pullback 
		\[H^1(M;\R)\to H^1(T;\R)\]
		via the map $g\mapsto g\cdot p $ is surjective.
	\end{lem}
	\begin{proof}
		The conclusion does not depend on $p$ since $M$ is connected. 
  
  Let $\alpha_1,\dots,\alpha_d$ be $T$-invariant closed 1-forms on $M$ such that $\alpha=\alpha_1\wedge\cdots\wedge\alpha_d$ is a volumne form. Let $p$ be an arbitary point of $M$, then $\alpha_1|_{T\cdot p},\dots,\alpha_d|_{T\cdot p}$ are $T$-invariant differential 1-forms that span the cotangent space at $p$, hence must span $H^1(T\cdot p;\R)$.
	\end{proof}
	
	\subsection*{Teleman's conjecture without potential}\label{TelConj}
	In the rest of this subsection, we assume $\mu^{-1}(0)\not \subset D$ and $T$ acts freely on $\mu^{-1}(0)$, in particular, $0\in \lt^*$ is a regular value of $\mu$. We let $Z$ denote the symplectic reduction $Y\sslash T=\mu^{-1}(0)/T$. We will show that $(\pi^\vee)^{-1}(1)\cong Z^\vee$, the SYZ mirror of $Z$ (see definitions below). 
	
	We first show that $Z$ has a natural meromorphic volume form. Let $\pr:\mu^{-1}(0)\to Z$ be the quotient map and $D_Z=\pr(\mu^{-1}(0)\cap D)$.
	
	\begin{lem}\label{holvolquo}
		The $(d-r,0)$-form 
		\[\Omega'\vcentcolon=\iota_{\lambda_r^\sharp}\cdots\iota_{\lambda_1^\sharp}\Omega\]
		descends to a meromorphic volume form $\Omega_Z$ on the K\"ahler manifold $Z$ and 
  \[
  \operatorname{div}(\Omega_Z)+D_Z\equiv 0.
  \]
	\end{lem}
	\begin{exmp}
		Suppose $Y=\C^{n+1}$ with complex coordinates $z_0,\dots,z_n$, and $\operatorname{U}(1)$ acts on $Y$ via
		\[t\cdot(z_0,\dots,z_n)=(tz_0,\dots,tz_n),\]
        and with a moment map $\mu=\sum_j|z_j|^2-c$ for some $c>0$. The fundamental vector field is $\lambda^\sharp=2\pi i\sum z_j\partial_{z_j}$.
		Consider the meromorphic volume form
		\[\Omega=\frac{dz_0\cdots dz_n}{z_0\cdots z_n}.\]
		Then
		\begin{equation*}
			\iota_{\lambda^\sharp}\Omega=2\pi i\sum_{j=0}^n(-1)^j\frac{dz_0\cdots\overset{\wedge}{dz_j}\cdots dz_n}{z_0\cdots \overset{\wedge}{z_j}\cdots z_n},
		\end{equation*}
		which descends to a meromorphic volume form on $\PP^n$ with a simple pole at each coordinate hyperplane.
	\end{exmp}
	\begin{proof}[Proof of Lemma \ref{holvolquo}]
		Recall that a differential form $\alpha$ on $\mu^{-1}(0)$ descends to $Z$ if and only it is basic, i.e., $\iota_{\xi^\sharp}\alpha=\mathcal{L}_{\xi^\sharp}\alpha=0$ for any $\xi\in \lt$.\newline
		We first assume $\mu^{-1}(0)\cap D=\varnothing$. Let $\xi\in \lt$, it is clear that $\iota_{\xi^\sharp}\Omega'=0$.\newline
		Since $[\mathcal{L}_{\xi^\sharp},\iota_{\lambda_j^\sharp}]=\mathcal{L}_{[\xi^\sharp,\lambda_j^\sharp]}=\mathcal{L}_{-[\xi,\lambda_j]^\sharp}=0$ for $j=1.\dots,r$, we have
		\[\mathcal{L}_{\xi^\sharp}\Omega'=\iota_{\lambda_r^\sharp}\cdots\iota_{\lambda_1^\sharp}\mathcal{L}_{\lambda^\sharp}\Omega=0\]
		by the $T$-invariance of $\Omega$.\newline
		We next check that $\Omega_Z$ is non-degenerate. Let $p\in \mu^{-1}(0)$ and $\bar p=\pr(p)$. Let $\bar v_{r+1},\dots, \bar v_d\in T^{1,0}_{\bar p}Z$ be a basis, and $v_{r+1},\dots,v_d\in T^{1,0}_pY$ be some liftings, then
		\begin{align*}
			\Omega_Z(\bar v_{r+1},\dots,\bar v_d)&=\Omega(\lambda_1^\sharp,\dots,\lambda_r^\sharp,v_{r+1},\dots,v_d)\\
			&=\Omega\left(\frac{1}{2}(\lambda_1^\sharp-iJ\lambda_1^\sharp),\dots,\frac{1}{2}(\lambda_r^\sharp-iJ\lambda_r^\sharp),v_{r+1},\dots,v_d\right)\\
			&\neq 0.
		\end{align*}
		It remains to check that $d\Omega'=0$. In fact,
		\[d\Omega'=(-1)^r\iota_{\lambda_r^\sharp}\cdots\iota_{\lambda_1^\sharp}d\Omega+\sum_{j=1}^r(-1)^{r-j}\iota_{\lambda_r^\sharp}\cdots\overset{\wedge}{\iota}_{\lambda_j^\sharp}\cdots\iota_{\lambda_1^\sharp}\mathcal{L}_{\lambda_j^\sharp}\Omega=0.\]
		We now consider the general case. The above argument shows that $\Omega_Z$ is well-defined and nonvanishing on $\pr(\mu^{-1}(0)\setminus D)$.\newline
		Let $p\in \mu^{-1}(0)\cap D$, and let $f$ be a local $T$-invariant defining equation of $D$ near $p$. Apply the above argument, we see that $f\Omega$ descends to a nonvanishing holomorphic volume form in a neighborhood of $\bar p=\pr(p)$. But $f$ also descends to a holomorphic function in a neighborhood of $\bar p$ and is thus a local equation of $D_Z$.
	\end{proof}
	For each vector field $u$ on $Z$, we let $u^h$ denote the unique lift of $u$ to $\mu^{-1}(0)$ that is orthogonal to the $T$-orbits. Similarly, if $S$ is a subset of $Z$ and $u$ is a vector field to $Z$ along $S$, we let $u^h$ denote the unique lift of $u$ to $\pr^{-1}(S)$ that is orthogonal to the $T$-orbits.\newline
	We write $\omega_Z$ for the K\"ahler form of $Z$ obtained from symplectic reduction.
	\begin{lem}\label{matchdef}
		Let $L\subset \mu^{-1}(0)$ be a special Lagrangian of $Y$ and $\bar L=\pr (L)$. Suppose $u\in C^\infty(N\bar L)$, then
		\begin{enumerate}[i)]
			\item $d\iota_u\omega_Z=0|_{\bar L}$ if and only if $d\iota_{u^h}\omega|_L=0$.
			\item $d\iota_u\operatorname{Im}(\Omega_Z)|_{\bar L}=0$ if and only if $d\iota_{u^h}\operatorname{Im}(\Omega)|_L=0$.
		\end{enumerate}
	\end{lem}
	\begin{proof}
		\begin{enumerate}[i)]
            \item[]
			\item It follows from the fact $\pi^*$ is injective, and the identity
			\begin{equation}\label{F-holo}
                \iota_{u^h}\omega=\iota_{u^h}\pr^*\omega_Z=\pr^*\iota_u\omega_Z.
			\end{equation} 
			\item We first observe the following fact: for each $g\in T$, $g_*u^h$ is also a lift of $u$ that is orthogonal to $\xi^\sharp$ for each $\xi\in \lt$. As a result, $g_*u^h=u^h$ and $[\mathcal{L}_{u^h},\iota_{\xi^\sharp}]=-\mathcal{L}_{[\xi^\sharp,u^h]}=0$ for each $\xi\in \lt$. Now,
			\begin{align*}
				\pr^*(d\iota_u\operatorname{Im}(\Omega_Z)|_{\bar L})&=\pr^*(\mathcal{L}_u\operatorname{Im}(\Omega_Z)|_{\bar L}   )\\&=\mathcal{L}_{u^h}\iota_{\lambda_r^\sharp}\cdots\iota_{\lambda_1^\sharp}\operatorname{Im}(\Omega)|_L\\
				&=\iota_{\lambda_r^\sharp}\cdots\iota_{\lambda_1^\sharp}\mathcal{L}_{u^h}\operatorname{Im}(\Omega)|_L,
			\end{align*}
   which vanishes if and only if $d\iota_{u^h}\operatorname{Im}(\Omega)|_L=\mathcal{L}_{u^h}\operatorname{Im}(\Omega)|_L=0$.
		\end{enumerate}
	\end{proof}
        We define $Z^\vee$ to be the set of $(\bar L, \bar E)$ where $\bar L\subset Z\setminus D_Z$ is a special Lagrangian torus with respect to $\Omega_Z$, and $\bar E$ is a flat $\operatorname{U}(1)$ connection on $\bar L$. $Z^\vee$ has a K\"ahler structure and has a nonvanishing holomorphic volume form.
        On the other hand, since we assume $0$ is a regular value of $\mu^\vee$, \Cref{criticaliff} implies $1\in 
        \check T_\C$ is a regular value $\pi^\vee$. We equip $(\pi^\vee)^{-1}(1)$ with the holomorphic volume form $\Omega_1$ so that $(2\pi)^r\Omega_1\wedge (\pi^\vee)^*\Omega_{\check T_\C}= \Omega_{Y^\vee}^\vee$ along $(\pi^\vee)^{-1}(1)$.
	
	\begin{prop}\label{quotient=fibre}
		$Z^\vee$ is isomorphic to $(\pi^\vee)^{-1}(1)\subset Y^\vee$ as a Calabi-Yau manifold.
	\end{prop}
	
	\begin{proof}
		Let $\bar L\subset Z$ be a Lagrangian torus, then $L=\pr^{-1}(\bar{L})$ is a $T$-invariant Lagrangian torus of $Y$. This gives a bijection between Lagrangian tori of $Z$ and $T$-invariant Lagrangian tori of $Y$ contained in $\mu^{-1}(0)$. 
		
		Let $p\in L$ and $\bar p=\pr(p)\in \bar L$. For $v_{r+1},\dots,v_d\in T_pL$ and $\bar{v}_{j}=d\pr(u_j)$, we have
		\[\Omega_Z(\bar v_{r+1},\dots,\bar v_d)=\Omega(\lambda_1^\sharp,\dots,\lambda_r^\sharp,v_{r+1},\dots,v_d).\]
		Hence $\overline{L}$ is a special Lagrangian of $Z$ (with respect to $\Omega_Z$) if and only if $L$ is a special Lagrangian of $Y$. Moreover, a flat $\operatorname{U}(1)$ connection $E$ on $L$ is equal to the pullback of a connection $\bar E$ on $L$ if and only if $\operatorname{hol}_T(E)\equiv 1$.

        We define $F:Z^\vee\to Y^\vee$ by $F(\bar L,E)=(\pr^{-1}(\bar L),\pr^*E)$. It follows from the above that $F$ gives a bijection $Z\cong (\pi^\vee)^{-1}(1)$. 
		
		If $(u,\alpha)\in T_{(\bar L,\bar E)}Z^\vee$, then $dF(u,\alpha)=(u^h,\pr^* \alpha)$. In particular, $F$ is smooth by \Cref{matchdef}, and is a diffeomorphism by the inverse function theorem. The holomophicity of $F$ follows from the definitions and \Cref{F-holo}.
		
		We now check the compatibility of the K\"ahler forms. Let $u_1, u_2\in C^\infty (N{\bar L})$ be such that $d\iota_{u_j}\Omega_Z=d\iota_{u_j}\operatorname{Im}(\Omega_Z)=0$ for $j=1,2$. We claim that $\iota_{u_j}(\Omega_Z)=\pr_*(\iota_{u_j^h}\Omega)$ on $\bar L$. In fact, let $p\in L$, $\bar p=\pr(p)$, and $v_{r+2}\dots,v_{d}\in T_{\bar p}\bar L$, we consider the parametrization $g:[0,1]^r\to T\cdot p$ sending $t=(t_1,\dots,t_r)$ to $\exp(t_1\lambda_1+\cdots+t_r\lambda_r)p$, then
		\begin{align*}
			&(\pr_*\iota_{u_j^h}\Omega)_{\bar p}(v_{r+2},\dots,v_{d})\\
			=&\int_0^1\cdots \int_0^1\Omega(\lambda^\sharp_1,\cdots,\lambda_r^\sharp,u_j^h,v_{r+2}^h,\dots,v_{d}^h)_{g(t)\cdot p}dt_1\cdots dt_r\\
			=&\int_0^1\cdots \int_0^1\pr^*\Omega_Z(u_j,v_{r+2},\dots,v_{d})dt_1\cdots dt_r\\
			=&\Omega_Z(u_j,v_{r+2},\dots,v_{d}).
		\end{align*}
		Let $\alpha_1,\alpha_2\in H^1(\bar L)$, we have
		\begin{align*}
			\omega^\vee_{Z^\vee}((u_1,\alpha_1),(u_2,\alpha_2))=&\int_{\bar L}\alpha_2\wedge\iota_{u_1}\operatorname{Im}(\Omega_Z)-\alpha_1\wedge\iota_{u_2}\operatorname{Im}(\Omega_Z)\\
			=&\int_{\bar L}\alpha_2\wedge\pr_*(\iota_{u_1^h}\operatorname{Im}(\Omega))-\alpha_1\wedge\pr_*(\iota_{u_2^h}\operatorname{Im}(\Omega))\\
			=&\int_{ L}\pr^*\alpha_2\wedge\iota_{u_1^h}\operatorname{Im}(\Omega)-\pr^*\alpha_1\wedge\iota_{u_2^h}\operatorname{Im}(\Omega)\\
			=&\omega^\vee_{Y^\vee}((u_1^h,\pr^*\alpha_1),(u_2^h,\pr^*\alpha_2))
		\end{align*}
		as claimed.\newline
		Finally, we compare the holomorphic volume forms. Let $u_j\in C^{\infty}(N{\bar L})$ be such that $\iota_{u_j}\Omega_Z=\iota_{u_j}\operatorname{Im}(\Omega_Z)=0$ and $\alpha_j\in H^1(L)$ for $j=r+1,\dots,n$. For each $j=1,\dots,r$, we also choose a vector field $v_j$ along $L$ so that $\iota_{v_j}\omega$ is $\psi_L$-harmonic and pullbacks to $d\ell_j$ along the orbit map $T\to T\cdot p\subset L$. We have
		\begin{align*}
			&\Omega^\vee_{Z^\vee}(u_{r+1},\alpha_{r+1}),\dots,(u_d,\alpha_d)\\
			=&\int_{\bar{L}}((-\iota_{u_{r+1}}\omega_Z+i\alpha_{r+1})\wedge\cdots\wedge (-\iota_{u_d}\omega_Z+i\alpha_d))\\
			=&\int_L \iota_{v_1}\omega\wedge\cdots \wedge\iota_{v_r}\omega\wedge(-\iota_{u^h_{r+1}}\omega+i\pr^*\alpha_{r+1})\wedge\cdots\wedge (-\iota_{u^h_d}\omega+i\pr^*\alpha_d)\\
   =&\Omega^\vee_{Y^\vee}((v_1,0),\dots,(v_r,0),(u^h_{r+1},\pr^*\alpha_{r+1}),\dots,(u^h_d,\pr^*\alpha_d)).
		\end{align*}

    It suffices to verify that $\Omega_{\check T_\C} (\pi^\vee_*(v_1,0),\dots,\pi^\vee_*(v_r,0))=(2\pi )^{-r}$. Note that $2\pi i\ell_1,\dots,2\pi i\ell_r$ forms an integral basis for $\Lambda_{\check T}$ and we have $\langle 2\pi i\ell_j,\frac{\lambda_k}{2\pi i}\rangle=\delta_{jk}$. Therefore,
        \[\Omega_{\check T_\C}=\frac{1}{(2\pi )^{r}}\bigwedge_{j=1}^r d\lambda_j.\]
    On the other hand, we see from \Cref{critical} that $d\lambda_j(d \mu^\vee (v_k,0))=-\langle \ell_k,\lambda_j\rangle=-\delta_{jk}$, the result then follows.
	\end{proof}
	\begin{rem}
		There are two superpotentials defined on $Z^\vee$: one is obtained from $Z$ using (\ref{Wsum}), while the other is obtained by restricting the superpotential of $Y^\vee$ to $(\pi^\vee)^{-1}(1)\cong Z$. Teleman \cite{tel2014} checked in some cases that these two superpotentials agree. In the joint work of the second author with Lau and Li \cite{LLL}, the method of (equivariant) Lagrangian correspondence-trimodules is adopted to study this problem, and the difference between the two superpotentials is accounted for by counting the (equivariant) disc potential of the moment-level Lagrangian.
	\end{rem}
	%
	
	\subsection{Equivariant SYZ type transform relative to a complex torus}\label{fibration over complex tori}
    Throughout this subsection, we let $T^n$ be a compact torus of dimension $n$ with Lie algebra $\lt$, and $\rho:T\to T^n$ be a homomorphism. We also fix a $T$-equivariant holomorphic map 
    \[
    \pi:Y\to T^n_\C.
    \]
    We will show that there is a Hamiltonian $\check T^n$ action on $Y^\vee$ so that $\pi^\vee:Y^\vee\to \check T_\C$ is equivariant with respect to the dual homomorphism $\check \rho:\check T^n \to \check T$.

    Note that $\check T^n$ parametrized the unitary flat line bundles on $T^n_\C$, the $\check T^n$ action on $Y^\vee$ is defined by tensoring with the pullback of line bundles on $T_\C$.

    More explicitly, let $\uL=(L,E)\in Y^\vee$, and $\nabla_E$ denote the flat connection on $E$. For any $i\ell\in \Hom_\Z(\pi_1(T^n_\C),i\R)\cong i H^1(T^n_\C;\R)$, we choose a closed 1-form $\beta_\ell$ on $T^n_\C$ so that $i\beta_\ell$ represents the class of $i\ell$. The action of $\check T^n$ on $Y^\vee$ is defined as
    \[\exp(i\ell)\cdot (L,E)=(L,E'),\]
    where
    \[\nabla_{E'}=\nabla_{E}-i\pi^*\beta_\ell|_L.\]
	\begin{prop}\label{G-hamaction}
		The action of $\ \check T^n$ on $Y^\vee$ described above preserves the complex structure $J^\vee$ and is Hamiltonian with a moment map given by
		\begin{equation}\label{check pi}
			\mu^\vee(L,E)=-\int_L\log|\pi|\operatorname{Re}(\Omega).
		\end{equation}
	Here, we define $\log |g\exp (\xi)|=\xi$ for $g\in T^n$ and $i\xi\in \lt^n$,
	\end{prop}

\begin{proof}
If we view $g\in\check T^n$ as a map $Y^\vee\to Y^\vee$, then $(dg)_{\uL}:T_{\uL}Y^\vee\to T_{g\cdot \uL}Y^\vee$ can be identified canonically with the identity map on $H^1(L;\C)$, and is thus complex linear.

Let $(u,\alpha)\in T_{\uL}Y^\vee$, and $i\ell\in \check\lt^n =i(\lt^n)^*$. If $\mu^\vee$ is defined by the formula (\ref{check pi}), then
		\begin{align*}
			\langle  i\ell,d\mu^\vee (u,\alpha)\rangle&=-\int_L\mathcal{L}_{u}\left (\langle i\ell,\log|\pi|\rangle\operatorname{Re}(\Omega)\right)\\
			&=-\int_L\langle i\ell,d\log|\pi|\rangle(u)\operatorname{Re}(\Omega)+\int_L\langle i\ell,d\log|\pi|\rangle\wedge\iota_u\operatorname{Re}(\Omega).
		\end{align*}

We write $a=\langle i\ell,d\log|\pi|\rangle$, and $b=-a\circ J_Y$. Since $\pi$ is holomorphic,
		\begin{equation*}
			b(\xi^\sharp)=\ell(\xi)
		\end{equation*}
For any $\xi\in \lt^n$. Therefore,
		\begin{align*}
			\check \omega((u,\alpha),(i\ell)^\sharp)&=\int_Lb\wedge \iota_u\operatorname{Im}(\Omega).
		\end{align*}
  It remains to verify
  \begin{equation*}
      b\wedge \iota_u\operatorname{Im}(\Omega)=-b(J_{T^n_\C} u)\operatorname{Re}(\Omega)+b\circ J_{T^n_\C} \wedge\iota_u\operatorname{Re}(\Omega)
  \end{equation*}
 on $L$, which follows from \Cref{lemmavolform} below.
 \end{proof}
	
	\begin{lem}\label{lemmavolform}
		Let $(z_1,\dots,z_d)$ be the complex coordinates of $\ \C^d\cong \R^{2d}$ with $z_j=x_j+iy_j$, and $\Omega=dz_1\wedge\cdots\wedge dz_d$. Let $\alpha=\sum_j(a_jdx_j+b_jdy_j)$ and $u=\sum_j(p_j\frac{\partial}{\partial x_j}+q_j\frac{\partial}{\partial y_j})$.
		Then
		\[\alpha\wedge \iota_u\operatorname{Im}(\Omega)|_L=-\alpha(Ju)\operatorname{Re}(\Omega)|_L+\alpha\circ J\wedge \iota_u\operatorname{Re}(\Omega)|_L=\sum_ja_jq_jdx_1\wedge\cdots\wedge dx_n\]
		where $L$ is the linear subspace of $\R^{2d}$ defined by $y_1=y_2=\cdots=y_d=0$
	\end{lem}
	\begin{proof}
		It follows from a direct calculation.
	\end{proof}
    \begin{prop}\label{TT'eq}
        The holomorphic map $\mu^\vee:Y^\vee\to \check T_\C$ is $\check T^n$-equivariant.
    \end{prop}

    \begin{proof}
    Let $g \in \check{T}^n$, and denote by $E_g$ the corresponding unitary flat line bundle on $T^n_\C$, i.e.,  $\hol_T(E_g) = \check{\rho}(g)$.
    
    Now, let $\uL=(L,E)\in Y^\vee$, then $g\cdot \uL= (L,E')$, where $E'=E\otimes \pi^*E_g|_L$. Therefore, 
        \[\hol_T(E')=\check \rho(g)\hol_T(E).\]
        The result then follows from \Cref{formulapi}.
    \end{proof}
	\begin{prop}\label{W0Ginv}
 For each disc class $\beta\in H_2(Y,L)$, the function $z_\beta:Y^\vee\to \C$,
 \[z_\beta:(L,E)\mapsto \exp(-\int_\beta\omega)\hol_{\partial \beta}(E)\]
 is $\check T^n$-invariant. In particular, if the sum (\ref{Wsum}) converges absolutely, then $W:Y^\vee\to \C$ is $\check T^n$-invariant.
	\end{prop}
	\begin{proof}
		Let $\beta\in H_2(Y,L)$, then the image of $\partial \beta$ under $\pi$ is null homotopic, thus any unitary flat line bundle on $T^n_\C$ pullbacks to a trivial line bundle on $\partial \beta$. The result then follows from the formula of $z_\beta$ above.
	\end{proof}
 
	\subsection{Equivariant fibrations over a toric variety}\label{fibration over toric varieties}
 In the remaining part of this subsection, we let $\PP_\Sigma$ denote an $n$ dimensional smooth toric variety, with dense torus $T^n_\C\cong (\Cx)^n$ and fan $\Sigma$. We also let $\rho:T\to T^n$ be a homomorphism, therefore, $T$ acts on $\PP_\Sigma$.

Let $\Sigma (1)=\{u_{1},u_{2},\dots ,u_{l}\}$ be the generators of the rays
in $\Sigma $. We write $H_{j}$ for the toric divisor in $\mathbb{P}_{\Sigma }$
corresponding to $u_{j}$.

We assume there is a $T$-equivariant holomorphic map $\pi :Y\rightarrow \mathbb{P}_{\Sigma }$ such that the anti-canonical divisor $D$ of $Y$ has a decomposition
\begin{equation*}
D=\sum_{j=0}^{l}D_{j},
\end{equation*}%
where $D_{j}=\pi ^{-1}(H_{j})$ (counted with multiplicities) for $j>0$ and $%
D_{0}$ has no irreducible components contained in $\cup _{j}D_{j}$. This
implies that $\pi (Y\setminus D)\subset T_{\C}^{n}$, the results in 
\Cref{fibration over
complex tori} provide a Hamiltonian $\check{T}^{n}$-action on $Y^{\vee }$.

 \begin{prop}\label{equivariant W}
 Let $\beta\in H_2(Y,L)$ be a disc class such that $\beta\cdot D_j=n_j$ for $j=1,\dots,l$, and $z_\beta:Y^\vee\to \C$ be the function
 \[\uL=(L,E)\mapsto \exp(-\int_\beta\omega)\hol_{\partial \beta}(E).\]
 Then for each $g\in \check T^n$ and $\uL=(L,E)\in Y^\vee$, we have
 \[z_\beta(g
 \uL)=\left(\prod_{i=1}^lu_j(g)^{n_j}\right) z_\beta(\uL).\]
 \end{prop}

 \begin{proof}
    Suppose $g\in \check T^n$ corresponds to a unitary flat line bundle $E_g$ on $T^n_\C$, then $z_\beta(g\uL)=\hol_{\pi_*(\partial \beta)}(E_g)z_\beta(p)$. However, $\pi_*(\partial \beta)$ is homotopic to the product loop $S^1\to \PP_\Sigma$, $z\mapsto \prod_{j=1}^lu_{j}^{n_j}(z)\cdot p$, where $p\in L$ is an arbitrary point. Since $\hol_{u_j(S^1)\cdot p}(E_g)=u_j(g)$, we have
    \[\hol_{\pi_*(\partial \beta)}(E_g)=\prod_{i=1}^lu_j(g)^{n_j}\]
    as claimed.	
	\end{proof}
 We can understand \Cref{equivariant W} as follows. Suppose the sum (\ref{Wsum}) converges absolutely, then we can write
 \[W=\sum_{j=0}^lW_j,\]
 where
 \begin{equation}\label{Wjdef}
			W_j(L,E)= \sum_{\substack{\beta\in H_2(Y,L)\\\beta\cdot D=\beta\cdot D_j=1}}n_\beta(L)z_\beta
		\end{equation}
  for $j=0,1,\dots,l$.
	
 \Cref{equivariant W} says that $W_0$ is $\check T^n$-invariant, and the holomorphic map 
  \[W_{\Sigma(1)}=(W_1,W_2,\dots,W_l):Y^\vee\to \C^{\Sigma(1)}\cong \C^l\] 
  is $\check T^n$-equivariant, where $\check T^n_\C$ acts on $\C^l$ via the homomorphism
  \[
  (u_1,\dots,u_l):T^n_\C\to(\Cx)^l. 
  \]

We can summarize the results of this section in the following theorem. 
\begin{thm}\label{SYZtransform}
        Let $Y$ be a K\"ahler manifold with a Hamiltonian $T$-action, $\rho:T\to T^n$ be a homomorphism, and $\PP_\Sigma$ be an $n$ dimensional smooth toric variety with dense torus $T_\C^n$ and fan $\Sigma$. Suppose $\Sigma(1)=\{u_1,u_2,\dots,u_l\}$, and $H_i\subset \PP_\Sigma$ be the toric divisor corresponding to $u_i$. Assume 
        \begin{enumerate}[(a)]
            \item there is a $T$-equivariant holomorphic map $\pi:Y\to \PP_\Sigma$;
            \item $Y$ has a $T$-invariant anti-canonical divisor 
        \[D=\sum_{j=0}^lD_j\]
        such that $D_i$ is the inverse images of the toric divisors $H_1$ of $\PP_\Sigma$, and $\pi(D_0)$ does not contain any toric divisors;
        \item The sum 
        \begin{equation*}
			\sum_{\substack{\beta\in H_2(Y,L)\\\beta\cdot D=1}}n_\beta(L)\exp (-\int_\beta\omega)\operatorname{hol}_{\partial \beta}(E)
		\end{equation*} 
  converges absolutely.
        \end{enumerate}
         Then
        \begin{enumerate}[(i)]
            \item there is a holomorphic map $\pi^\vee:Y^\vee\to \check T_\C$ defined by 
            \begin{equation*}
			\pi^\vee(L,E)\mapsto \exp(-\mu(L))\cdot \hol_T(E);
		\end{equation*}
            \item there is a holomorphic and Hamiltonian $\check T^n$-action on $Y^\vee$, with moment map $\mu^\vee:Y^\vee\to \lt^*$ given by 
            \begin{equation*}
			\mu^\vee(L,E)=-\int_L\log|\pi|\operatorname{Re}(\Omega).
		\end{equation*}
            \item if we write $W=\sum_{j=0}^lW_j$, where $W_j$ is defined by
        \begin{equation*}
			W_j:(L,E)\mapsto \sum_{\substack{\beta\in H_2(Y,L)\\\beta\cdot D=\beta\cdot D_j=1}}n_\beta(L)z_\beta,
		\end{equation*} then $W_0$ is $\check T^n$-invariant, and the holomorphic map 
            \[F=W_{\Sigma(1)}\times \pi^\vee=(W_1,W_2,\dots,W_l)\times \pi^\vee:Y^\vee\to\C^{\Sigma(1)}\times \check T_\C\] 
            is $\check T^n$-equivariant.
        \end{enumerate}
\end{thm}
The assumptions of the theorem say that $[Y \sslash T] \to C$ is a 3d brane of $X=T^*C$, while the conclusions assert that $[Y \sslash T] \to C^!$ is a 3d brane of $X^!=T^*C^!$. Hence, \Cref{SYZtransform} offers a method to construct a mirror 3d brane, transforming the 3d brane $(C, [Y \sslash T])$ of $T^*C$ into the mirror 3d brane $(C^!, [Y^\vee \sslash \check{T}^n])$ of $T^*C^!$.

\section{K\"ahler and equivariant parameters}\label{matchingKahEqu}
Given a complex symplectic manifold $X$,  let $\lt_{X,\C}$ be the Lie algebra
of the maximal torus of $\text{Ham}\left( X\right) $ and $H_{X,\C}^{2}=H^{2}\left( X,
\mathbb{C}\right) $. $H_{X,\C}^{2}$ and $\lt_{X,\C}$ are called the spaces of K\"ahler
parameters and equivariant parameters of $X$ respectively. 3d mirror
symmetry predicts that 
\begin{eqnarray*}
&&%
\begin{array}{ccc}
H_{X,\C}^{2} & \cong  & \lt_{X^{!},\C}%
\end{array}
\\
&&%
\begin{array}{ccc}
\lt_{X,\C} & \cong  & H_{X^{!},\C}^{2},
\end{array}%
\end{eqnarray*}%
where $X$ is the 3d mirror of $X$.
Note that in 3d mirror symmetry, the complex symplectic form is usually
exact and $H^{2}\left( X\right) =H^{1,1}\left( X\right) $. For
example this is the case for conical symplectic resolution $X$.

\begin{exmp} When $X=T^*\mathbb{P}^{n-1}$, then $X^!=[T^*\C^n\sslash (\Cx)^{n-1}]$, where $(\Cx)^{n-1}$ is the subtorus of $(\Cx)^n$ with product $1$. Note that there is a Hamiltonian action of $T_{X^!,\C}\cong (\Cx)^n/(\Cx)^{n-1}$ on $X^!$. Therefore, $H^2_{X,\C}\cong \C\cong\lt_{X^!,\C}$. One can also check that $\lt_{X,\C} \cong  H_{X^{!},\C}^{2}$.
\end{exmp}

Analogous to our description of the 3d mirror symmetry between $X$ and $X^{!}
$ via SYZ perspective, we will explain the exchange of K\"ahler and
equivariant parameters in 3d mirror symmetry via the exchange of
symplectic and complex parameters in mirror symmetry. 

Recall a 3d brane of $X$ is a diagram $Y\overset{\pi }{\rightarrow }C\overset{\iota }{\rightarrow }X$, where $Y$ is a K\"ahler
manifold with the complexified K\"ahler class $\left[ \omega _{Y}\right] $, 
$C\overset{\iota }{\rightarrow }X$ is a complex Lagrangian submanifold and $%
\pi :Y\rightarrow C$ is a holomorphic map.  
\[
\begin{array}{ccc}
Y &  &  \\ 
\downarrow  &  &  \\ 
C & \overset{\iota }{\rightarrow } & X.
\end{array}
\]%

We call $[\omega_Y]$ the symplectic structure of the 3d brane, while $J_Y$ and the holomorphic $\pi \circ \iota: Y \to X$ are refered to as its complex structure.

We first notice that there is a natural action of $\text{Ham}\left( X\right) $ on 
$\mathcal{B}r\left( X\right) $ given as follow. For any $g\in \text{Ham}(X)$, $g\cdot \left( Y\overset{\pi }{%
\rightarrow }C\overset{\iota }{\rightarrow }X\right) $ is defined to be $Y%
\overset{\pi }{\rightarrow }C\overset{g\circ \iota }{\rightarrow }X$.
Namely, the action is given by moving $C$ inside $X$ by $g$. 

On the other hand, given any complexified K\"ahler form $\alpha$ on $X$, we obtain another 3d brane in $X$ by changing the complexified K\"ahler class on $Y$ to $\left[  \omega _{Y}+\pi ^{\ast }\iota ^{\ast }\alpha \right] $. This gives a
semi-group action of the complexifixed K\"ahler cone $\mathcal{K}_{X}$ on $%
\mathcal{B}r\left( X\right) $. It is clear that the $\text{Ham}\left( X\right) $%
-action and $\mathcal{K}_{X}$-action on $\mathcal{B}r\left( X\right) $
commune with each other. That is%
\[
\text{Ham}\left( X\right) \times \mathcal{K}_{X}\curvearrowright \mathcal{B}%
r\left( X\right).
\]
It is evident that $\text{Ham}(X)$ deforms the complex structures of 3d branes, while $\mathcal{K}_X$ deforms their symplectic structures.

When $X$ and $X^{!}$ are 3d mirror to each other, $\mathcal{B}r\left(
X\right) $ and $\mathcal{B}r\left( X^{!}\right) $ are related by mirror
symmetry, or SYZ type transformations. We expect the interchange of symplectic structures on $Y$ and complex structures on its mirror $Y^{\vee }$ will give us the identification of the infinitesimal action of $H_{X}$ on $\mathcal{B}r\left(
X\right) $ and and $t_{X^{!}}$ on $\mathcal{B}r\left( X^{!}\right)$. In the following, we will verify this in the hypertoric case using the SYZ type transform developed in \Cref{Aurouxpicture}.

\subsection{The case of hypertoric varieties/stacks}
In this subsection, we focus on the case where $X$ and $X^!$ are hypertoric varieties or stacks, and demonstrate that the exchange of K\"ahler and equivariant parameters in 3d mirror symmetry corresponds, at the level of 3d branes, to the exchange of symplectic and complex structures in 2d mirror symmetry.

We follow the notations introduced in \Cref{Functoriality1,Aurouxpicture}. Specifically, we have
\[
C = [\PP_\Sigma / T_\C], \quad C^! = [\C^l \times \check{T}^n_\C].
\]

Let $X=T^*C$, and $X^!=T^*C^!$, and set
\[
\begin{tabular}{c c}
\\
     $H^2_{X,\C}=H^2_T(\PP_\Sigma)$
     &$\lt_{X,\C}=\lt^n_\C/\im(d\rho)$  \\
     \\
     \\
    $ H^2_{X^!,\C}=H^2_{\check T^n}(\C^l\times \check T_\C)$
     &
     $\lt_{X^!,\C}=(\C^l\oplus \check\lt_\C)/\check \lt^n_\C$.\\
     \\
\end{tabular}
\]
Let $Y$ be a $T$-Hamiltonian K\"ahler manifold, and $\pi: Y\to \PP_\Sigma$ be a $T$-equivariant holomorphic map. In other words, $[Y\sslash T]\to C$ defined a 3d brane of $X$.
\begin{df}
    Let $\alpha=\alpha_2+\alpha_0$ be a $T^n$-invariant representative of a class in $H^2_{X,\C}=H^2_T(\PP_\Sigma)$. In otherwords, $\alpha_2$ is a $T^n$-invariant closed 2-form on $\PP_\Sigma$ and $\alpha_0:\PP_\Sigma\to \lt_\C^*$ is a $T^n$-invariant function such that $\langle d\alpha_0,\xi\rangle=\iota_{\xi^\sharp}\alpha_2$ for any $\xi\in\lt$. 
    Suppose $\omega_Y+\pi^*\alpha$ is a (complexified) K\"ahler form, then we define $[Y_\alpha\sslash T]\to C$ to the 3d brane of $X$ such that
    \begin{enumerate}
        \item $Y_\alpha=Y$ as a complex manifold.
        \item \label{YM12}$\omega_{Y_\alpha}=\omega_Y+\pi^*\alpha_2$ and $\mu_{Y_\alpha}=\mu_{Y}-\pi^*\alpha_0$.
        \item $\pi_\alpha:Y_\alpha\to \PP_\Sigma$ is equal to $\pi$.
    \end{enumerate}
\end{df}

This is the infinitesimal version of the $\mathcal{K}_X$-action introduced earlier.

\begin{df}
Let $\xi\in \lt^n_C$, then we define $[Y_\xi\sslash T]\to C$ to be the 3d brane of $X$ 
\begin{enumerate}
    \item $Y=Y_\xi$ as a $T$-Hamiltonian K\"ahler manifold.
    \item $\pi_\xi:Y_\xi\to \PP_\Sigma$ is given by
    \begin{equation}\label{pixi}
    \pi_\xi(L,E)=\exp(\xi)\pi(L,E).
    \end{equation} 
\end{enumerate}   
\end{df}
This is the restriction of the $\text{Ham}(X)$-action on $\Br(X)$ introduced earlier to $\lt_{X^!,\C}$.

We will relate the two actions via the transform of 3d branes. First let $[\alpha]\in H^2_{X^!,\C}$, and suppose the assumptions of \Cref{SYZtransform} are satisfied for the 3d branes $[Y\sslash T]\to C$ and $[Y_\alpha\sslash T]\to C$, then there are $\check T^n$-equivariant holomorphic map
\begin{align*}
    F:Y^\vee&\to \C^l\times \check T_\C\\
    F_\alpha:Y_\alpha^\vee&\to \C^l\times \check T_\C.
\end{align*}
Suppose $L$ is a Lagrangian for both $Y$ and $Y_\alpha$, and $E$ be a unitary flat line bundle on $L$.

Recall $F=(W_1,\dots,W_l,\pi^\vee)$ is defined by
\begin{align}
    W_j(L,E)&=\sum_{\beta\cap D_i=\delta_{ij}}n_\beta(L)\exp (\int_\beta-\omega_Y)\hol_{\partial \beta}(E)\label{WY}\\
    \pi^\vee(L,E)&=e^{-\mu_Y(L)}\hol_{T}(E).\label{piY}
\end{align}
Similarly, $F_\alpha=(W_{1,\alpha},\dots,W_{l,\alpha},\pi^\vee_\alpha)$ is defined by
\begin{align}
    W_{j,\alpha}(L,E)&=\sum_{\beta\cap D_i=\delta_{ij}}n_\beta(L)\exp (\int_\beta-\omega_{Y_\alpha})\hol_{\partial \beta}(E)\label{WYalpha}\\
    \pi^\vee(L,E)&=e^{-\mu_{Y_\alpha}(L)}\hol_{T}(E).\label{piYalpha}
\end{align}
    
Let $q\in \pi(L)$, and for $i=1,\dots,l$, let $\beta_{j,q}\in H_2(\PP_\Sigma;T^n\cdot q)$ be the unique disc class with $\beta_{j,q}\cap H_i=\delta_{ij}$. From (\ref{WY}) and $(\ref{WYalpha})$, we have
\begin{align*}
    W_{j,\alpha}(L,E)&\vcentcolon= \sum_{\beta\cdot H_j=1}n_\beta(L)\exp (\int_\beta-\omega-\alpha_2)\hol_{\partial \beta}(E)\\
    &=\exp(\int_{\beta_{j,q}}-\alpha_2)W_j(L,E),
\end{align*}
From (\ref{piY}) and (\ref{piYalpha}), we also have
\begin{equation*}
    \pi^\vee_\alpha(\uL)=\exp(\alpha_0(p))\pi^\vee(\uL).
\end{equation*}

To summarize, if we define
\begin{align}\label{discformula}
    \Phi:H^2_{T}(\PP_\Sigma;\C)&\to \lt_{X^!,\C}=\frac{\C^l\times \check \lt_\C}{\check \lt^n_\C}\\
    \alpha&\mapsto\left(\int_{\beta_{1,p}}\alpha_2,\dots,\int_{\beta_{l,p}}\alpha_2,-\alpha_0(p)\right).
\end{align}
Then 
\[
F_\alpha(L,E)=\exp(-\Phi(\alpha))F(L,E).
\]

\begin{rem}
    One can check that $\Phi$ is independent of $q$. In a forthcoming paper, we will prove a generalization of this assertion when $C$ is not necessarily toric. By making use of the fact the areas of holomorphic discs are positive, we will also justify that $C$ is semistable with respect to a K\"ahler class $\omega\in H^2_X$ if and only if $C^!$ is attracting with respect to $\Phi(\omega)$.
\end{rem}

We next consider deformations of $\pi$ by an equivariant parameter. Let $\xi\in \lt^n_C$. Recall that the moment maps $\mu^\vee:Y^\vee\to \check\lt^*$ and $\mu_\xi^\vee:Y_\xi^\vee\to \check\lt^*$ are given by

\begin{align*}
    \mu^\vee(L,E)&=-\int_L\log|\pi|\operatorname{Re}(\Omega)\\
    \mu_\xi^\vee(L,E)&=-\int_L\log|\pi_\xi|\operatorname{Re}(\Omega)
\end{align*}

Note that $\int_L\operatorname{Re}(\Omega)$ is a locally constant function on $Y^\vee$, and we assume it is equal to $1$ for simplicity. Then it is straightforward to see
\begin{align*}
    \check{\omega}-\mu^\vee_\xi&=\check{\omega}-\mu^\vee+\xi.
\end{align*}
In other words, we obtain a deformation of the symplectic structure on $Y^\vee$ via the pullback $F^* \Psi(\xi)$, where
\begin{align}\label{Psimap}
    \Psi: \lt_{X,\C} \to H^2_{\check{T}^n}(\C^l \times \check{T}_\C)
\end{align}
maps $\xi \in \lt_{X,\C}$ to the cohomology class represented by the constant function on $\C^l \times \check{T}_\C$ with value $\xi$. The following theorem summarizes the calculations made so far.

\begin{thm}\label{thmparameters}
Let $\pi:Y\to \PP_\Sigma$ be a $T$-equivariant holomorphic map satisfying the assumptions of \Cref{SYZtransform}. Let $\uL=(L,E)\in Y^\vee$ be such that $\pi(L)$ is contained in a $T^n$-orbit in $\PP_\Sigma$, and $\int_L\operatorname{Re}(\Omega)=1$. Then
    \begin{enumerate}
        \item for any $\alpha\in H^2_{X,\C}$,
        \[
        F_\alpha(\uL)=\exp(-\Phi([\alpha]))F(\uL).
        \]
        Here, $\Phi:H^2_{X^,\C}\to \lt_{X^!,\C}$ is defined in (\ref{discformula}); 
        \item for any $[\xi]\in \lt_{X,\C}$, 
        \[
        [\check{\omega}-\mu^\vee_\xi]=[\check{\omega}-\mu^\vee]+F^*\Psi([\xi]).
        \] 
        Here, $\Psi:\lt_{X,\C}\to H^2_{X^!,\C}$ is defined in (\ref{Psimap}).    
    \end{enumerate}
\end{thm}
To conclude, \Cref{thmparameters} says that the exchange of K\"ahler and equivariant parameters in 3d mirror symmetry corresponds, in the brane level, to the exchange of symplectic and complex structures in 2d mirror symmetry. 

\begin{exmp}
    Let $X = T^* \PP^{n-1}$ and $X^! = [T^* \C^n // (\Cx)^{n-1}]$. If we take $Y = \PP_\Sigma \to \PP_\Sigma$ to be the identity map, then
    \begin{equation*}
    \begin{tikzcd}
        Y\ar[r,equal]&\mathbb{P}^{n-1}\ar[d,"\text{id}"]\\
        C\ar[r,equal]&\mathbb{P}^{n-1}
    \end{tikzcd}
    \text{ is 3d mirror to }
    \begin{tikzcd}
        (\Cx)^{n-1}\ar[r,symbol=\curvearrowleft]\ar[d,"F"]&(\Cx)^{n-1}\ar[d,equal]\\
        \C^n\ar[r,symbol=\curvearrowleft]&(\Cx)^{n-1},
    \end{tikzcd}
\end{equation*}
Here, $F$ is equal to
\[\left(x_1,x_2,\dots,x_{n-1},\frac{q}{\prod_{i=1}^{n-1} x_i}\right),\]
    where $q\in \Cx$ depends on the K\"ahler class of $Y$. This shows that varying the K\"ahler form of $Y$ by a K\"ahler parameter of $X$ corresponds to rotating $F$ by the $T_{X^!,\C}$-action (equivariant parameter).
\end{exmp}

\section{Applications of complex Lagrangian correspondences}\label{Gluing}

\subsection{Complex Lagrangian correspondences}

Given a complex Lagrangian submanifold $\mathcal{L}$ in the product of
complex symplectic manifolds $X_{1}\times X_{2}$, if $C_{1}$ is a complex
Lagrangian submanifold in $X_{1}$, we denote 
\begin{eqnarray*}
Y_{2} &=&\pi _{1}^{-1}\left( C_{1}\right) \cap \mathcal{L\subset ~}%
X_{1}\times X_{2} \\
C_{2} &=&\pi _{2}\left( Y_{2}\right) \subset X_{2}
\end{eqnarray*}
where $\pi_{i}: X_{1} \times X_{2} \rightarrow X_{i}$ is the projection to the $i^{\text{th}}$-factor, with $i=1,2$. If the intersection $\pi_{1}^{-1}(C_{1}) \cap \mathcal{L}$ is transverse, then $\pi_{2}|_{Y_2}: Y_{2} \rightarrow X_{2}$ is a complex Lagrangian immersion. If the intersection $\pi_{1}^{-1}(C_{1}) \cap \mathcal{L}$ is clean, then $Y_{2} \rightarrow C_{2}$ is a holomorphic fibration, which defines a 3d brane in $X_{2}$. This construction can be directly generalized by (i) replacing $C_{1}$ with a 3d brane $Y_{1} \rightarrow C_{1}$ in $X_{1}$, and (ii) replacing $\mathcal{L}$ with a 3d brane in $X_{1} \times X_{2}$.

Namely we have 
\begin{eqnarray*}
\Phi  &:&\mathfrak{Br}\left( X_{1}\times X_{2}\right) \times _{clean}%
\mathfrak{Br}\left( X_{1}\right) \rightarrow \mathfrak{Br}\left(
X_{2}\right)  \\
\left( \mathcal{L},Y_{1}\rightarrow C_{1}\right)  &\mapsto &\Phi _{\mathcal{L%
}}\left( Y_{1}\rightarrow C_{1}\right) =\left( Y_{2}\rightarrow C_{2}\right). 
\end{eqnarray*}%
Here $\times _{clean}$ refers to the clean asssumption mentioned above. $%
\Phi _{\mathcal{L}}$ is called the complex Lagrangian correspondence from $%
X_{1}$ to $X_{2}$, induced from $\mathcal{L}\subset X_{1}\times X_{2}$.
Without the clean assumption, we could still define $\Phi $ in the realm of
symplectic geometry \cite{ShiftSym,LagCor}. 

It is expected that 3d branes will generate a 2-category in both the A-model and the B-model, and that such a $\Phi_{\mathcal{L}}$ should induce a 2-functor between these 2-categories for $X_{1}$ and $X_{2}$.

\subsection{Product structures}\label{Monoidal} 

Complex Lagrangian correspondences can similarly be defined for stacks. For example, the diagonal homomorphism $T_{\C} \rightarrow T_{\C} \times T_{\C}$ induces a map of stacks
\begin{equation*}
\Delta : [\text{pt}/T_\C] \rightarrow [\text{pt}/T_\C] \times [\text{pt}/T_\C].
\end{equation*}
The conormal to the graph of $\Delta$ is a complex Lagrangian $\mathcal{L}_{\Delta}$ in $T^{\ast}[\text{pt}/T_{\C}] \times T^{\ast}[\text{pt}/T_{\C}] \times T^{\ast}[\text{pt}/T_{\C}]$. This induces a product structure on $\mathfrak{Br}_{A}(T^*[\pt/T_\C])$ as follows.

Let $Y_{1}$ and $Y_{2}$ be two $T$-Hamiltonian manifolds with moment maps $\mu_{Y_{1}}: Y_{1} \rightarrow \mathfrak{t}^{\vee}$ and $\mu_{Y_{2}}: Y_{2} \rightarrow \mathfrak{t}^{\vee}$.
The diagonal action on $Y_{1} \times Y_{2}$ is also Hamiltonian, with moment map $\mu_{Y_{1} \times Y_{2}}(y_{1}, y_{2}) = \mu_{Y_{1}}(y_{1}) + \mu_{Y_{2}}(y_{2})$. The product structure
\begin{equation}
\Phi_{\mathcal{L}_{\Delta}}: \mathfrak{Br}_{A}\left( T^{\ast}[\text{pt}/T_{\C}] \right) \times \mathfrak{Br}_{A}\left( T^{\ast}[\text{pt}/T_{\C}] \right) \rightarrow \mathfrak{Br}_{A}\left( T^{\ast}[\text{pt}/T_{\C}] \right)
\label{Ldelta}
\end{equation}
on $\mathfrak{Br}_{A}(T^{\ast}[\text{pt}/T_{\C}])$ is defined by 
\begin{equation}\label{Ldelta2}
    ([Y_1\sslash T],[Y_2\sslash T])\mapsto [Y_1\times Y_2\sslash T].
\end{equation}
On the 2-categorical level, we also expect $\mathcal{L}_{\Delta}$ would induce a monoidal structure on the 2-category for the A-model on $T^{\ast}[\text{pt}/T_{\C}]$, which is compatible with $\Phi_{\mathcal{L}_{\Delta}}$ on the object level.

\begin{rem}
Suppose the $T$-action on $Y_1$ and $Y_2$ extends to $T_\C$-actions. Then $\Phi_{\mathcal{L}_{\Delta}}([Y_1/T_\C],[Y_2/T_\C])$ can be computed by the pullback diagram
\[
\begin{tikzcd}
    \Phi_{\mathcal{L}_{\Delta}}(\left[Y_1/T_\C\right],\left[Y_2,T_\C\right])\ar[r]\ar[d]& \left[Y_1/T_\C\right]\times \left[Y_2/T_\C\right]\ar[d]\\
    \left[\pt/T_\C\right]\ar[r,"\Delta"]&\left[\pt/T_\C\right]\times \left[\pt/T_\C\right].
\end{tikzcd}
\]
Therefore $\Phi_{\mathcal{L}_{\Delta}}([Y_1/T_\C])=[Y_1\times Y_2/T_\C]$. In general, we define $\Phi_{\mathcal{L}_{\Delta}}$ by (\ref{Ldelta2}).
\end{rem}

\bigskip 

The 3d mirror of $T^{\ast}[\text{pt}/T_{\C}]$ is $T^{\ast} \check{T}_{\C} = \check{T}_{\C} \times \mathfrak{t}_{\C}$. The conormal to the graph of the group product
\begin{equation*}
m: \check{T}_{\C} \times \check{T}_{\C} \rightarrow \check{T}_{\C}
\end{equation*}
is a complex Lagrangian submanifold in $T^{\ast} \check{T}_{\C} \times T^{\ast} \check{T}_{\C} \times T^{\ast} \check{T}_{\C}$. Explicitly,
\begin{equation*}
\mathcal{L}_{m} = \{ (g_{1}, h, g_{2}, h, g_{1}g_{2}, h) \mid g_{1}, g_{2} \in \check{T}_{\C}, h \in \mathfrak{t}_{\C} \} \subset (\check{T}_{\C} \times \mathfrak{t}_{\C})^{3},
\end{equation*}
which induces a product structure
\begin{equation*}
\Phi_{\mathcal{L}_{m}}: \mathfrak{Br}_{B}(T^{\ast} \check{T}_{\C}) \times \mathfrak{Br}_{B}(T^{\ast} \check{T}_{\C}) \rightarrow \mathfrak{Br}_{B}(T^{\ast} \check{T}_{\C}).
\end{equation*}
This corresponds to a monoidal structure for $T^{\ast} \check{T}_{\C}$, which is the 3d mirror to the one above for $T^{\ast}[\text{pt}/T_{\C}]$. The complex Lagrangian $\mathcal{L}_{m}$ is part of a symplectic groupoid structure on $T^{\ast} \check{T}_{\C}$ (\cite{symgpoid, Izu}), so the complex Lagrangian correspondence $\Phi_{\mathcal{L}_{m}}$ is associative\footnote{A similar symplectic groupoid structure can be defined on $T^{\ast}G$ for a non-Abelian group $G$, but this is not the focus of this paper.}, which can also be verified directly.

Let $C_{1}, C_{2} \subset T^{\ast} \check{T}_{\C} = \check{T}_{\C} \times \mathfrak{t}_{\C}^{\ast}$ be two complex Lagrangian submanifolds. Then $C_{1} \times C_{2}$ is a complex Lagrangian submanifold in $T^{\ast} \check{T}_{\C} \times T^{\ast} \check{T}_{\C}$. Applying the complex Lagrangian correspondence $\mathcal{L}_{m}$ to $C_{1} \times C_{2}$ yields the complex Lagrangian
\begin{equation*}
C_{1} \star_{m} C_{2} = \Phi_{\mathcal{L}_{m}} \left( C_{1} \times C_{2} \right) = \{ (g_{1} g_{2}, h) : (g_{1}, h) \in C_{1}, (g_{2}, h) \in C_{2} \} \subset T^{\ast} \check{T}_{\C}.
\end{equation*}
In other words, $C_{1} \star_{m} C_{2}$ is the product of $C_{1}$ and $C_{2}$ using the group scheme structure of $T^{\ast} \check{T}_{\C} \rightarrow \mathfrak{t}_{\C}$.

We now verify that $\Phi_{\mathcal{L}_{\Delta}}$ and $\Phi_{\mathcal{L}_{m}}$ are 3d mirrors of each other. In other words, the following diagram commutes:
\begin{equation*}
\begin{array}{ccc}
\mathfrak{Br}_{A}(T^{\ast}[\text{pt}/T_{\C}]) \times \mathfrak{Br}_{A}(T^{\ast}[\text{pt}/T_{\C}]) & \xleftrightarrow{\text{3d mirror symmetry}} & \mathfrak{Br}_{B}(T^*T_\C) \times \mathfrak{Br}_{B}(T^*T_\C) \\
\downarrow \Phi_{\mathcal{L}_{\Delta}} & & \downarrow \Phi_{\mathcal{L}_{m}} \\
\mathfrak{Br}_{A}(T^{\ast}[\text{pt}/T_{\C}]) & \xleftrightarrow{\text{3d mirror symmetry}} & \mathfrak{Br}_{B}(T^*T_\C)
\end{array}
\end{equation*}

Let $Y$ be a $T$-Hamiltonian manifold with mirror $(Y^\vee, W)$ and Teleman map $\pi^\vee: Y^\vee \to \check{T}_{\C}$. In other words, the 3d brane $(Y^\vee, W)/\check T_\C$ of $T^{\ast} \check{T}_{\C}$ is the mirror of the 3d brane $[Y \sslash  T] / [\text{pt} / T_{\C}]$ of $T^{\ast} [\text{pt} / T_{\C}]$. Recall from \Cref{nonabelian} that we use $W$ to ``push" the 3d brane $(Y^\vee, W) / \check{T}_{\C}$ out of the zero section, resulting in
\[
C_Y = \{ (g, h) \in \check{T}_{\C} \times \mathfrak{t}_{\C} \mid g = \pi_i(p), \, df_i(p) = \pi_i^{\ast} dh \text{ for some } p \in Y \}.
\]

Now, for $i = 1, 2$, let $Y_i$ be a $T$-Hamiltonian manifold with mirrors $(Y_i^\vee, W_i)$ and Teleman maps $\pi_i^\vee: Y_i^\vee \to \check{T}_{\C}$. Let $Y = Y_1 \times Y_2$ be equipped with the diagonal $T$-action. Suppose $(Y^\vee, W) = (Y_1^\vee \times Y_2^\vee, W_1 + W_2)$ is the mirror of $Y$, and $\pi^\vee(y_1, y_2) = \pi_1^\vee(y_1) \cdot \pi_2^\vee(y_2)$ is the corresponding Teleman map. The following lemma shows that $C_Y = C_{Y_1} \star_m C_{Y_2}$.

\begin{lem}
\label{lemmaProduct}
For $i = 1, 2$, let $M_i$ be complex manifolds, $f_i$ be holomorphic functions on $M_i$, and $\pi_i: M_i \to \check{T}_{\C}$ be holomorphic maps. Define $M = M_1 \times M_2$, with the functions $f: M \to \mathbb{C}$ given by $f(m_1, m_2) = f_1(m_1) + f_2(m_2)$, and $\pi: M \to \check{T}_{\C}$ given by $\pi(m_1, m_2) = \pi_1(m_1) \cdot \pi_2(m_2)$. If we set
\begin{equation*}
C_i = \{(g, h) \in \check{T}_{\C} \times \mathfrak{t}_{\C} : g = \pi_i(p), \, df_i(p) = \pi_i^*dh \text{ for some } p \in M_i\}
\end{equation*}
for $i = 1, 2$, and 
\begin{equation}\label{dfC}
C = \{(g, h) \in \check{T}_{\C} \times \mathfrak{t}_{\C} : g = \pi(p), \, df(p) = \pi^*dh \text{ for some } p \in M\},
\end{equation}
then $C = C_1 \star_m C_2$.
\end{lem}

\begin{proof}
Let $m = (m_1, m_2) \in M$ and $h \in \mathfrak{t}_{\C}$. Then $df(m) = \pi^*(dh)$ if and only if $df_i(m_i) = \pi_i^*(dh)$ for $i = 1, 2$.
\end{proof}

\begin{exmp}
Let $Y = \text{pt}$ with the trivial action and $\mu(\text{pt}) = 0 \in \mathfrak{t}^\vee$. Here, $Y$ serves as the unit for the product structure on $\mathfrak{Br}_A(T^*[\text{pt}/T_\C])$. On the other hand, $Y^\vee = \text{pt}$, and $\pi^\vee:Y^\vee \to \check{T}_\C$ is the constant map with value $1$. We have
\begin{equation*}
C = \{1\} \times \mathfrak{t}_\C \subset \check{T}_\C \times \mathfrak{t}_\C,
\end{equation*}
which is also the unit for the symplectic groupoid structure on $T^*\check{T}_\C$.
\end{exmp}

\begin{exmp}\label{exmpC}
\begin{enumerate}
    \item []
    \item Let $\rho$ be a character of $T$, and $Y = \mathbb{C}_\rho$ the corresponding one-dimensional representation. We can take $(Y^\vee, W) = (\mathbb{C}^\times, z)$ as the Hori-Vafa mirror, with $\pi^\vee = \check \rho : \mathbb{C}^\times \to \check{T}_\C$. Then
    \begin{equation*}
    C = \{(g, h) \in \check{T}_\C \times \mathfrak{t}_\C : g = \check \rho (d\rho(h))\}.
    \end{equation*}

    \item More generally, let $\rho_1, \dots, \rho_n$ be characters of $T_\C$, and $Y = \prod_{j=1}^n \mathbb{C}_{\rho_j}$. Then
    \begin{equation*}
    C = \{(g, h) \in \check{T}_\C \times \mathfrak{t}_\C : g = \prod_{j=1}^n \rho_j^\vee(d\rho_j(h))\}.
    \end{equation*}
\end{enumerate}
\end{exmp}

\begin{rem}\label{antisym}
    If we replace the mirror $(Y^\vee, W)$ of $T \curvearrowright Y$ with $(Y^\vee, -W)$, then $C_Y$ will be replaced by 
    \[
    -C_Y = \{(g, h) \in \check{T}_\C \times \mathfrak{t}_\C \mid (g, -h) \in C\}.
    \]
    In other words, $-C_Y$ is the image of $C_Y$ under the anti-symplectomorphism of $T^*\check{T}_\C$ given by $(g, h) \mapsto (g, -h)$.
\end{rem}

\begin{rem}\label{BFM}
    For a non-Abelian compact group $G$, the 3d mirror of $T^*[\pt/G_\C]$ is the BFM space $\text{BFM}(\check G)$ discussed in \Cref{nonabelian}. By a result in \cite{BFM}, $\text{BFM}(\check G)$ can be identified with the universal centralizer of $\check{G}_\C$, and thus possesses a symplectic groupoid structure. The discussion in this subsection extends to the non-Abelian setting using the methods of \Cref{nonabelian}. See also \Cref{nonrem}.
\end{rem}

\subsection{Coulomb branch of 3d gauge theory with matters}\label{glueing section}

3d $\mathcal N=4$ gauge theory with matters refers to the 3d theory for $X=[T^{\ast
}V\sslash G_\C]$ with $V$ being a complex representation of $G$. Braverman,
Finkelberg and Nakajima \cite{BFN} described a method to construct the
Coulomb branch for such a theory, which should be the 3d mirror $X^{!}$ to $X
$. We will assume $G=T$ is Abelian, suppose the $T$-action on $V$ is
faithful, then $X$ is a hypertoric variety (if one chooses a K\"ahler parameter, otherwise, it is a stack) and $X^{!}$ is its Gale dual
hypertoric variety. On the other extreme situation when $V=0$, then $%
X=T^{\ast }[$pt$/T_\C]$ and $X^{!}=T^{\ast }\check T_{\C}$. See \Cref{nonrem} for remarks on the non-Abelian case.

In this subsection, we will utilize the complex Lagrangian correspondences $\Phi_{\mathcal{L}_{\Delta}}$ for $T^{\ast}[\text{pt}/T_\C]$ and $\Phi_{\mathcal{L}_{m}}$ for $T^{\ast}\check{T}_\C$ to construct gluing maps, thereby obtaining the 3d mirror $X^{!}$ to $X = [T^{\ast}V\sslash T_\C]$. This process involves gluing various local charts $T^{\ast}\check{T}_\C$, with each chart corresponding to an identification $T^{\ast}V \simeq T^{\ast}V_{I}$ for some $T$-representation $V_{I}$. The affine variety version of the gluing will give the Coulomb branch described in \cite{BFN}.

Explicitly, if we fix an irreducible decomposition 
\begin{equation*}
V = \bigoplus_{i=1}^{n} U_{i}
\end{equation*}
where each $U_{i}$ is a one-dimensional representation of $T$ corresponding to some character $\rho_i$, then for any decomposition $\left\{ 1,2,\dots, n \right\} = I \bigsqcup I^{\prime}$, we can write
\begin{equation*}
V_{I} = \left( \oplus_{i \in I} U_{i} \right) \oplus \left( \oplus_{i \in I^{\prime}} U_{i} \right)^{\ast},
\end{equation*}
where $V_{I} \subset T^{\ast}V$ is a Lagrangian subspace, and $T^{\ast}V \simeq T^{\ast}V_{I}$ with a symplectic form defined by suitably chosen signs for each term.

Each $\left[ V_{I}/T \right]$ is a complex Lagrangian subspace in $\text{pt} \times X$ and $X \times \text{pt}$ for $X=T^*[V/T_\C]$, thus inducing complex Lagrangian correspondences: 
\begin{align*}
    F_I&:\Br_A(T^*[\pt/T_\C])\to\Br_A(X), \text{ and}\\
    G_I&:\Br_A(X)\to\Br_A(T^*[\pt/T_\C])
\end{align*}
such that 
\[
G_{I^{\prime}} \circ F_{I} = \text{id}.
\]
Explicitly, given $T \curvearrowright \left( Y, \omega \right) \overset{\mu}{\rightarrow} \mathfrak{t}^{\ast}$, namely a 3d A-brane $(\left[ \text{pt}/T_\C \right],\left[ Y\sslash T \right] )$, its image under $F_{I}$ is the following 3d A-brane in $X$:
\begin{equation*}
F_{I}\left( \left[ \text{pt}/T_\C \right], \left[ Y\sslash T \right] \right) = (\left[ V_{I}/T_\C \right],\left[ Y \times V_{I} \sslash  T \right]),
\end{equation*}
namely, taking the product with $V_{I}$. For $G_{I}$, given $\left( Z, \omega \right) \rightarrow C \subset T^{\ast}V$ which is $T$-equivariant.
\begin{equation*}
G_{I}\left(  \left[ C / T_\C \right],\left[ Z \sslash  T \right] \right) =( \left[ \text{pt}/T_\C \right],\left[ Z \times_{T^{\ast}V} V_{I'} \sslash  T \right]  )
\end{equation*}
as a 3d A-brane in $T^{\ast} \left[ \text{pt}/T_\C \right]$.

For each $I \subset \left\{ 1, 2, \dots, n \right\}$, the 3d mirror of $F_{I}$ (resp. $G_{I^{\prime}}$) should be the pushforward (resp. pullback) of an open embedding $\iota_{I}: T^{\ast}\check{T}_\C \rightarrow X^{!}$, as we will explain below.

\begin{rem}
The above assertion can be verified in the following way. We assume $I=\{1,2,\dots,n\}$ and $\rho:T_\C\rightarrow (\mathbb{C}^{\times})^{n}\subset V$ is injective for simplicity. Let $Y$ be a $T$-Hamiltonian manifold, with mirror $\pi^\vee:Y^\vee\rightarrow \check{T}_\C$. Then, $F_{I}(Y)$ would be $Y\times V^{\ast}\rightarrow V^{\ast}$, where $T$ acts on $Y\times V^{\ast}$ via the diagonal action. The mirror of $Y\times V^{\ast}\rightarrow V^{\ast}$ is 
\begin{equation*}
\begin{tikzcd}
    Y^\vee\times (\mathbb{C}^\times)^n \ar{d}{(\text{pr}_2, \pi^\vee\cdot \check \rho^{-1})}/(\mathbb{C}^\times)^n\\
    \mathbb{C}^n\times \check{T}_\C/ (\mathbb{C}^\times)^n.
\end{tikzcd}
\end{equation*}
As $\rho$ is injective, we can use the $(\mathbb{C}^{\times})^{n}$-action to make the second factor become $1$, so the above diagram is isomorphic to 
\begin{equation*}
\begin{tikzcd}
    \{\pi^\vee\cdot\check{\rho}^{-1}=1\}/\ker\check{\rho} \ar{d}{\text{pr}_2}\\
    \mathbb{C}^n/ \ker\check{\rho},
\end{tikzcd}
\end{equation*}
which can be further simplified to 
\begin{equation*}
\begin{tikzcd}
    Y^\vee \ar{d}{\pi^\vee}\\
    \check{T}_\C \cong (\mathbb{C}^\times)^n/\ker\check{\rho} \subset \mathbb{C}^n/ \ker\check{\rho}.
\end{tikzcd}
\end{equation*}
Therefore, this mirror of $Y\times V^{\ast}\rightarrow V^{\ast}$ is the pushforward of the mirror of $Y$ under the inclusion $\check{T}_\C \subset \mathbb{C}^{n}/\ker \check{\rho}$ (or $T^{\ast}\check{T}_\C \subset T^{\ast}\mathbb{C}^{n}\sslash \ker \check{\rho}$).
\end{rem}

In particular, for $I, J \subset \{1,2,\dots,n\}$, correspondence $G_{J'} \circ F_{I}$ is 3d mirror to the change of coordinates $``\iota_J^{-1} \circ \iota_I"$:
\begin{equation*}
\begin{tikzcd}
    T^*[\text{pt}/T_\C] \ar[r, leftrightarrow, "\text{3d MS}"] \ar[d, dashrightarrow, "G_{J'} \circ F_I"] & T^*\check{T}_\C \ar[d, dashrightarrow, "\iota_J^{-1} \circ \iota_I"] \\
    T^*[\text{pt}/T_\C] \ar[r, leftrightarrow, "\text{3d MS}"] &
    T^*\check{T}_\C
\end{tikzcd}
\end{equation*}
We will explain below how to obtain the gluing map $\phi_{J,I} = \iota_J^{-1} \circ \iota_I$ by computing $G_{J'} \circ F_{I}$. One can check that 
\begin{equation*}
G_{J'} \circ F_{I}\left([\pt/T_\C],[Y\sslash T]\right) = ([\pt/T_\C],[Y \times V_{J,I}\sslash T]),
\end{equation*}
where
\[
V_{J,I} = \bigoplus_{i \in I \setminus J} U_i \oplus \bigoplus_{i \in J \setminus I} U_i^*.
\]
That is 
\begin{equation*}
G_{J'} \circ F_{I} = \Phi_{\mathcal{L}_{\Delta}}\left(([\pt/T_\C],[V_{J,I}/T_\C]), -\right) : \mathfrak{Br}_{A}\left(T^*[ \text{pt}/T_\C]\right) \rightarrow \mathfrak{Br}_{A}\left(T^*[ \text{pt}/T_\C]\right).
\end{equation*}
According to our studies in \Cref{Monoidal}, this is 3d mirror to
\begin{equation*}
\Phi_{\mathcal{L}_{m}}\left(C_{J,I}, -\right) : \mathfrak{Br}_{A}\left(T^* \check T_\C\right) \rightarrow \mathfrak{Br}_{A}\left(T^* \check T_\C\right),
\end{equation*}
where $C_{J,I}$ is the mirror complex Lagrangian of $T \curvearrowright V_{J,I}$. 

To calculate $C_{J,I}$, let $V_{J,I}^{\vee}$ be $(\mathbb{C}^{\times})^{J \setminus I} \times (\mathbb{C}^{\times})^{I \setminus J}$ with coordinates $z_{i}$ for $i \in J \setminus I$ and $z_{i}'$ for $i \in I \setminus J$. We consider the mirror of $V_{I \setminus J}$ as $V_{J,I}^{\vee}$ with the superpotential 
\[
\sum_{i \in J \setminus I} z_i - \sum_{i \in I \setminus J} z_i'.
\]

\begin{rem}\label{rem on twisted potential}
    The following is the reason why we should take the mirror of $V^*$ to be $((\Cx)^n, -\sum z'_i)$. Note that the isomorphism $T^*V \cong T^*V^*$ is anti-symplectomorphic, and we compensate for this by adding a minus sign to the potential function, which corresponds to an anti-symplectomorphism of $T^*\check{T}_\C$ (see \Cref{antisym}). Similar remark applies to the mirror of $V_{J,I}$.
\end{rem}
The following lemma generalizes \Cref{exmpC}, and allows us to calculate $C_{J,I}$.

\begin{lem}
\label{YM(W)} 
Let $\rho_1, \dots, \rho_n$ be characters of $T$, and $\epsilon_i \in \{1, -1\}$ for $i = 1, 2, \dots, n$. Suppose $\pi: M=(\mathbb{C}^\times)^n \to \check{T}_\C$ is the map 
\begin{equation*}
\pi(z_1, \dots, z_n) \mapsto \prod_{i=1}^{n} \rho_i^{\vee}(z_i),
\end{equation*}
and $f: M \to \mathbb{C}$ is the holomorphic function 
\begin{equation*}
f(z_1, \dots, z_n) \mapsto \sum_{i=1}^{n} \epsilon_i z_i.
\end{equation*}
If 
\begin{equation*}
C = \{(g, h) \in \check{T}_{\C} \times \mathfrak{t}_{\C} : g = \pi(p), \, df(p) = \pi^*dh \text{ for some } p \in M\},
\end{equation*}
is the complex Lagrangian of $T^*\check{T}_\C$ as in (\ref{dfC}), then

\begin{enumerate}
\item if $\rho_i$ is trivial for some $i$, then $C$ is empty; and

\item if $\rho_i$ is nontrivial for all $i$, then the Lagrangian correspondence $\Phi_{\mathcal{L}_{m}}\left(C, -\right)$ is given by a birational map $\phi_{M}:T^*\check T_\C\dashrightarrow T^*\check T_\C$. More precisely, let $f=f_M$ be the composition  
\begin{equation*}
f:\mathfrak{t}_\C \xrightarrow{d\rho} \mathbb{C}^n \xrightarrow{\epsilon} \mathbb{C}^n \dashrightarrow (\mathbb{C}^\times)^n \xrightarrow{\check \rho } \check{T}_\C,
\end{equation*}
where $\rho=(\rho_1,\dots,\rho_n)$, the second map is the multiplication of $\epsilon_i$ in the $i^{th}$ coordinate, and the third map is the identity rational map. Then $\Phi_{\mathcal{L}_{m}}\left(C, -\right)$ is given by the birational map  
\begin{align*}
\phi_{M}: \check{T}_\C \times \mathfrak{t}_\C & \to \check{T}_\C \times \mathfrak{t}_\C \\
(g, h) & \mapsto (f(h)g, h).
\end{align*}
\end{enumerate}
\end{lem}

\begin{proof}
 It follows from \Cref{lemmaProduct}, \Cref{exmpC}, and \Cref{antisym}.
\end{proof}

\begin{rem}\label{remformula}
In the function ring level, $\phi_M^*:\C[T^*T_\C]\to\C[T^*T_\C]$ is given by
\begin{equation*}
\phi_M^*:h \mapsto h; \quad z^{\lambda^{\vee}} \mapsto z^{\lambda^{\vee}} \prod_{i \in I} (\epsilon_ih_{\rho_{i}})^{\langle \rho_{i}, \lambda^{\vee} \rangle},
\end{equation*}
where we identify $T^{\ast }\check{T}_\C$ with $\check{T}_\C \times \mathfrak{t}_\C$, $h \in \mathfrak{t}_\C$, and $\lambda^\vee$ is a character of $T^\vee_\C$.
\end{rem}

In our situation, $M = V^\vee_{J,I}$, and $\epsilon: V^\vee_{J,I} \to V^\vee_{J,I}$ is the identity for the coordinates $z_i$ ($i \in I \setminus J$), and negative of the identity for the coordinates $z_i'$ ($i \in J \setminus I$). We set $\phi_{J,I} = \phi_{V_{J,I}}$. 

We will now check that the set $\{\phi_{J,I}\}_{I,J}$ satisfies the cocycle condition. Let $I, J, K$ be subsets of $\{1,2, \dots, n\}$. By definition, the graph of $\phi_{K,J} \circ \phi_{J,I}$ is $C_{V_{K,J,I}}$, where 
\begin{equation*}
V_{K,J,I} =V_{K,J}\oplus V_{J,I}= \left(\bigoplus_{i \in K \setminus J} U_i \right) \oplus \left(\bigoplus_{i \in J \setminus K} \overline{U_i^*} \right) \oplus \left(\bigoplus_{i \in J \setminus I} U_i \right) \oplus \left(\bigoplus_{i \in I \setminus J} \overline{U_i^*} \right).
\end{equation*}
and the graph of $\phi_{K,I}$ is $C_{V_{K,I}}$, where 
\begin{equation*}
V_{K,I} = \left(\bigoplus_{i \in K \setminus I} U_i \right) \oplus \left(\bigoplus_{i \in I \setminus K} \overline{U_i^*} \right).
\end{equation*}
Here and below, $\epsilon$ is equal to $-1$ on the factors with an overline, and is equal to $1$ on other factors.

We see that 
\[
V_{K,J,I} = V_{K,I} \oplus N \oplus \overline{N^*},
\]
where 
\begin{equation*}
N = \bigoplus_{i \in ((K \cap I)\setminus J) \cup (J \setminus (I \cup K))} U_i.
\end{equation*}

\begin{lem}
\label{YM(V)}
Let $V, V'$ be $T$-representations without trivial factors. Then:
\begin{enumerate}
\item $C_{V \oplus V'} = C_V \star_m C_{V'}$.
\item $C_{\overline{V}^*} = \{ (g,h) \in \check{T}_\C \times \mathfrak{t}_\C \mid (g^{-1},h) \in C_V \}$.
\end{enumerate}
\end{lem}

\begin{proof}
\begin{enumerate}
\item []
\item The proof is the same as that of \Cref{lemmaProduct}.
\item We may assume $V = \mathbb{C}$. We write $T_\C = (\mathbb{C}^\times)^r$, and $V$ corresponds to the character $\rho: (z_1, \dots, z_r) \mapsto \prod z_i^{\rho_i}$ of $T_\C$. By a direct calculation,
\begin{equation*}
C_V = \{ (z,h) \in (\mathbb{C}^\times)^r \times \mathbb{C}^r \mid z_i = (\rho_i h_i)^{\rho_i} \}.
\end{equation*}
On the other hand, the $T$-representation $V^*$ is induced by the character $\rho^{-1}: (z_1, \dots, z_r) \mapsto \prod z_i^{-\rho_i}$ of $T$, hence
\begin{equation*}
C_{\overline{V}^*} = \{ (z,h) \in (\mathbb{C}^\times)^r \times \mathbb{C}^r \mid z_i = \prod \left(-(-\rho_i) h_i\right)^{-\rho_i} \}.
\end{equation*}
(ii) is now clear.
\end{enumerate}
\end{proof}
The above lemma immediately implies the following.
\begin{cor}
For $I, J, K \subset \{1, 2, \dots, n\}$ such that $V_{K,J,I}$ contains no trivial factors, then
\[
C_{V_{K,J,I}} = C_{K,J}\star_m C_{J,I}.
\]
In particular, we have
\[
\phi_{K,J}\circ \phi_{J,I} = \phi_{K,I}.
\]
\end{cor}

The results of this subsection are summarized in the following theorem:

\begin{thm}\label{gluingthm}
    Let 
    \[
    V = \bigoplus_{i=1}^{n} U_{i}
    \]
    be a decomposition of a complex $T$-representation $V$ into the sum of one-dimensional representations $U_i$'s with character $\rho_i: T_\C \to \Cx$. For $I,J \subset \{1, 2, \dots, n\}$, we define
    \[
    V_{J,I} = \bigoplus_{i \in I \setminus J} U_i \oplus \bigoplus_{j \in J \setminus I} U_j^*\subset T^*V.
    \]
    Consider 
    \[
    (V_{J,I}^\vee, W_{J,I}) = \left( (\Cx_z)^{I \setminus J} \times (\Cx_{z'})^{J \setminus I}, \quad \sum_{i \in I \setminus J} z_i - \sum_{i \in J \setminus I} z_i' \right)
    \]
    as a 2d mirror of $V_{J,I}$; and the Teleman map
    \[
    \rho_{J,I}^\vee=\prod_{i\in I\setminus J}\check\rho_{i}\prod_{j\in J\setminus I}(\check\rho_j)^{-1}: V_{J,I}^\vee \to \check T_\C.
    \]
    Let $C_{J,I}$ be the complex Lagrangian in $T^*T_\C$ obtained from the graph $\Gamma_{dW_{J,I}}$ of $dW_{J,I}$, as described in \Cref{nonabelian}. In other words, $C_{J,I}$ is obtained by applying the complex Lagrangian correspondence induced by the conormal to the graph of $\check \rho_{J,I}$ to $\Gamma_{dW_{J,I}}$. Then
    \begin{enumerate}
        \item if $\rho_i$ is trivial for some $i \in I \triangle J$, then $\Phi_{\mathcal{L}_m}(C_{J,I},-)$ is an empty Lagrangian correspondence. In fact, we have $C_{J,I} = \varnothing$;
        \item if $\rho_i$ is nontrivial for all $i \in I \triangle J$, then $\Phi_{\mathcal{L}_m}(C_{J,I},-)$ is the graph of a birational map
        \[
        \phi_{J,I}: T^*\check T_\C \to T^*\check T_\C.
        \]
        More precisely, let $f$ be the composition  
\begin{equation*}
\begin{tikzcd}
    \mathfrak{t}_\C\ar[r,"d\rho_{J,I}"]& \mathbb{C}^n\ar[r,"\epsilon"] & \mathbb{C}^n \ar[r,dotted,"\operatorname{id}"]&(\mathbb{C}^\times)^n \ar[r,"\check \rho_{J,I}"]& \check{T}_\C
\end{tikzcd}
\end{equation*}
as in \Cref{YM(W)}, then $\phi_{J,I}$ is defined by
\begin{align*}
\phi_{J,I}: \check{T}_\C \times \mathfrak{t}_\C & \to \check{T}_\C \times \mathfrak{t}_\C \\
(g, h) & \mapsto (f(h)g, h).
\end{align*}
        \end{enumerate}
        Moreover, the $\phi_{J,I}$'s satisfy the cocycle condition:
        \[
        \phi_{K,I} = \phi_{K,J} \circ \phi_{J,I}
        \]
        for $I, J, K \subset \{1, 2, \dots, n\}$ where $\phi_{J,I}$ and $\phi_{K,J}$ are all defined.
\end{thm}
\Cref{gluingthm} says that we may obtain a smooth open subvariety of $X^!$ by gluing copies of $T^*\check{T}_\C$.

\begin{df}\label{dfX!}
\begin{equation*}
{X^{\!}} = \left(\mathcal{U}_I\right)/\sim,
\end{equation*}
where each $\mathcal{U}_I\cong T^*\check T_\C$, and $p_I\in \mathcal{U}_I$ is identified with $p_J\in \mathcal{U}_J$ if $p_J = \phi_{J,I}(p_I)$.
\end{df}

\begin{exmp}
Suppose $T_\C$ acts trivially on $V$, then $X^{\!}$ is the disjoint union of $2^n$ copies of $T^*\check{T}_\C$.
\end{exmp}

In \Cref{Compare Gale}, we will show that $X^{\!}$ is an open dense subset of the Gale dual $X^!$ in the hypertoric case, namely when the action $T\curvearrowright V$ is faithful.

\bigskip 

Next, we define an affine version of the construction in \Cref{dfX!}, and show that the affine variety obtained agrees with the BFN Coulomb branch defined in \cite{BFN}. We write 
\begin{equation*}
V = \bigoplus_{i=1}^{n} U_{i}
\end{equation*}
as before. For each $I \subset \{1,2,\dots,n\}$, we define a field automorphism $\phi_{I}^{\ast}$ of $ \mathbb{C}(T^{\ast }\check{T}_\C)$ given by 
\begin{equation}\label{affinegluing}
\phi_I^*:h \mapsto h; \quad z^{\lambda^{\vee}} \mapsto z^{\lambda^{\vee}} \prod_{i \in I} (-h_{\rho_{i}})^{\langle \rho_{i}, \lambda^{\vee} \rangle},
\end{equation}
where we identify $T^{\ast }\check{T}_\C$ with $\check{T}_\C \times \mathfrak{t}_\C$, $h \in \mathfrak{t}_\C$, and $\lambda^\vee$ is a character of $T^\vee_\C$. As remarked already in \Cref{remformula}, $\phi_I^*$ is the pullback of $\phi_I$ in the function ring level.

\begin{df}
We define $A_{T,V}$ as the subalgebra of $\mathbb{C}(T^*\check T_\C)$ consisting of functions $f$ such that $(\phi_I^*)^{-1}(f) \in \mathbb{C}[T^*\check T_\C]$ for all subsets $I \subset \{1,2,\dots,n\}$.
\end{df}

In other words, $A_{T,V}$ is the intersection of the rings $\phi_I^*\left(\mathbb{C}[T^*\check T_\C]\right)$ for all $I$.

\begin{rem}
\item The formula in (\ref{affinegluing}) remains valid even when some of the $\rho_i$'s are zero.
\end{rem}

For each $\lambda \in \Lambda$, we define 
\begin{equation*}
Z^\lambda = z^\lambda\prod_{\langle \rho_i, \lambda \rangle > 0} h_{\rho_i}^{\langle \rho_i, \lambda \rangle} ,
\end{equation*}
where $\Lambda$ is the weight lattice of $\mathfrak{t}_\C$. The algebra $A_{T,V}$ is generated by these elements $Z^\lambda$'s, which are elements in $\text{Sym}^\bullet(\mathfrak{t}_\C^*)$. The multiplication rule for $Z^\lambda$ is given by 
\begin{equation*}
Z^\lambda Z^\mu = \prod_{i=1}^n h_{\rho_i}^{d(\langle \rho_i, \lambda \rangle, \langle \rho_i, \mu \rangle)} Z^{\lambda + \mu},
\end{equation*}
where 
\begin{equation*}
d(m,n) = 
\begin{cases}
0 & \text{if } m \text{ and } n \text{ have the same sign}, \\ 
\min\{|m|, |n|\} & \text{if } m \text{ and } n \text{ have different signs}.
\end{cases}
\end{equation*}
By a comparison with Theorem 4.1 of \cite{BFN}, one see that $A_{T,V}$ is exactly the coulomb branch algebra defined by \citeauthor{BFN} in the Abelian case.

The following lemma explains when we can obtain $A_{T,V}$ by considering only the charts corresponding to $I=\varnothing$ and $I=\{1,2,\dots,n\}$ (namely corresponding to $V$ and $V^*$).

\begin{lem}\label{lemTel}
\[
A_{T,V}=\C[T^*\check T_\C]\cap \phi_{\{1,\dots,n\}}^*(\C[T^*\check T_\C])
\]
provided that $\rho_i$ is not a negative power of $\rho_j$ for any $i,j\in\{1,\dots,n\}$, unless $\rho_i=\rho_j=1$.
\end{lem}

\begin{proof}
    Consider a function $f \in \mathbb{C}[T^*\check{T}_\C]$, which can be expressed as
    \[
    f = \sum_{\lambda} z^{\lambda^\vee} f_\lambda(h),
    \]
    where $f_\lambda(h) \in \text{Sym}^\bullet(\mathfrak{t}_\C^*)$.

    Note that
    \[
    (\phi^*_{\{1,\dots,n\}})^{-1} f = \sum_{\lambda} z^{\lambda^\vee} f_\lambda(h) \prod_{i=1}^n(- h_{\rho_i})^{-\langle \rho_i, \lambda^\vee \rangle}.
    \]
    The assumption ensures that there is no cancellation among the terms in the product. Consequently, we must have
    \[
    f_\lambda(h) \prod_{i=1}^n (-h_{\rho_i})^{-\langle \rho_i, \lambda^\vee \rangle} \in \mathbb{C}[T^*\check{T}_\C],
    \]
    which implies
    \[
    \prod_{\langle \rho_i, \lambda \rangle > 0} h_{\rho_i}^{\langle \rho_i, \lambda \rangle} \mid f_\lambda(h).
    \]
    In otherwords, $f$ is a linear combination of $Z^\lambda$, and hence in $A_{T,V}$. This completes the proof of the lemma.
\end{proof}

In \Cite{tel2021}, Teleman gave a formula similar to \Cref{lemTel} without the assumption on the characters. We can obtain his formula by applying the following trick. Let $T' = T \times S^1$, and extend $V$ to a $T'$-representation where $S^1$ acts on each $U_i$ with weight 1.

It is evident that $V$ satisfies the conditions of \Cref{lemTel} when considered as a $T'$ representation. Additionally, Proposition 3.18 in \cite{BFN} establishes that $A_{T,V}$ is the Hamiltonian reduction of $A_{T',V}$ by $\mathbb{C}^\times$ (where $\check T'_\C = \check T_\C \times \mathbb{C}^\times$). Combining this with \Cref{lemTel}, we obtain the following:

\begin{thm}\label{thm4II}
    Let $t$ be an auxiliary variable, and extend $\phi_I^*$ to an automorphism of $\mathbb{C}(T^*\check{T}_\C)[t]$ that fixes $t$. Define
    \[
    A_{T,V}^t = \mathbb{C}[T^*\check{T}_\C][t] \cap \phi_I^*\left(\mathbb{C}[T^*\check{T}_\C][t]\right).
    \]
    Let $A_t$ be the subring of $\mathbb{C}[T^*\check{T}_\C][t]$ consisting of functions $f=f(z,h,t)\in \mathbb{C}[T^*\check{T}_\C][t]$ such that
    \[
    f\left(z \prod_{i=1}^n (t + h_{\rho_i})^{-\langle \rho_i, \lambda^\vee \rangle}, h, t \right) \in \mathbb{C}[T^*\check{T}_\C][t].
    \]
    Then, $A_{T,V}$ is isomorphic to $A^t_{T,V}/t A^t_{T,V}$.
\end{thm}

This is essentially the Theorem 1 in \cite{tel2021} for the Abelian gauge theory.

\begin{rem}\label{nonrem}
    Suppose $G$ is a compact group with a maximal torus $T \subset G$, and let $V$ be a representation of $G$. Suppose we fix a $T$-decomposition $V = \oplus_i U_i$ as before. As explained in \Cref{BFM}, if $V_I$ is indeed a $G$-representation, then $\phi_I^*$ induces a map $\mathbb{C}(\text{BFM}(\check G)) \to \mathbb{C}(\text{BFM}(\check G))$, which provides a gluing map in the non-Abelian setting.

    In \cite{tel2021}, Teleman shows that the nonabelian version of \Cref{thm4II} remains valid.
\end{rem}

\appendix

\section{Comparison with Gale duality}\label{Compare Gale}
We assume the action of $T$ on $V\cong (\Cx)^n$ is faithful, i.e., we have a short exact sequence
\[1\to T\xrightarrow{\rho} T^n\xrightarrow{\eta} T'\to 1.\]
Dualizing the above sequence, we get
\[1\to \check T'\xrightarrow{\check \eta} \check T^n\xrightarrow{\check \rho} \check T\to 1.\]
We can regard $V^*\cong \C^n$ as a representation of $\check T'$. The hypertoric stack $X=[T^*V\sslash T_\C]$ is 3d mirror to $X^!=[T^*V^*\sslash \check T'_\C]$. In this subsection, we will show that $X^{\!}$ is a substack of $X^!$ of the same dimension.

Let $a_1,a_2,\dots,a_n$ be coordinates of $V$, and $b_1,b_2,\dots,b_n$ be coordiniates of $V^*$ so that the natural pairing $V\times V^*\to \C$ is given by $(a,b)\mapsto \sum a_ib_i$. We write $V=\C^n_a$ and $V^*=\C^n_b$.

We consider the complex symplectic form
\[\omega_\C=\sum da_idb_i\]
on $T^*V^*=\C^n_b\times \C^n_a$.

For each subset $I$ of $\{1,2\dots,n\}$, we will define an open embedding $\kappa_I: (\Cx_z)^n\times \C^n_\zeta\cong T^*(\Cx_z)^n\to T^*V^*$ by the relation
\[\zeta_i=a_ib_i\]
and
\[z_i=
\begin{cases}
        a_i\text{ if }i \in I\\
	b_i^{-1},\text{ if }i\not\in I.
\end{cases}\]
We have
\[\kappa_I^*(\omega_\C)=\sum_i\frac{dz_i}{z_i}d\zeta_i, \]
which means $\kappa_I$ gives a symplectomorphism from $T^*(\Cx_z)^n$ onto an open subset $ Z_I\vcentcolon=\im \kappa_I$ of $T^*V^*$. 
\begin{prop}
    $X^{\!}$ is isomorphic to the symplectic reduction by $\check T'_\C$ of the union of $Z_I$ for all $I\subset \{1,2\dots,n\}$.
\end{prop}
\begin{proof}
Let $I$ and $J$ be two subsets of $\{1,2,\dots,n\}$, it is easy to check that
\[\kappa_J^{-1}\kappa_I(\zeta_i)=\zeta_i,\]
and
\[\kappa_J^{-1}\kappa_I(z_i)=
\begin{cases}
	\zeta_i^{-1}z_i,\text{ if }i\in I\setminus J,\\
	\zeta_iz_i\text{ if }i \in J\setminus I,\\
	z_i\text{ otherwise}.
\end{cases}\]

We denote $\iota_I$ the composition
\[T^*\check T_\C\cong [Z_I\sslash \check T'_\C]\xrightarrow{[\kappa_I\sslash \check T'_\C]} [T^*V^*\sslash \check T"_\C].\]
We have
\[\iota_J^{-1}\iota_I(z,\zeta)=(\check\rho(\zeta_{J,I})z,\zeta),\]
where we regard $\zeta\in\lt_\C\subset \C^n$, and $\zeta_{J,I}\in \C^n$ is such that
\[(\zeta_{J,I})_i=\begin{cases}
	\zeta_i^{-1}&\text{ if }i\in I\setminus J\\
	\zeta_i&\text{ if }i\in J\setminus I\\
	1&\text{ otherwise}.
\end{cases}\]
Denote by $I\triangle J=I\setminus J\cup J\setminus I$ the symmetric difference of $I$ and $J$. Let $\gamma_{J,I}:(\Cx)^n\to (\Cx)^{I\triangle J}$ be the homomorphism given by
\[\gamma_{J,I}(z)_i=\begin{cases}
	z_i^{-1}&\text{ if }i\in I\setminus J\\
	z_i&\text{ if }i\in J\setminus I.
\end{cases}\]
The map $\zeta\mapsto \zeta_{J,I}$ is equal to the composition 
\[\lt_\C\xrightarrow{ d(\gamma_{J,I}\circ\rho)}\C^{I\triangle J}\xrightarrow{\epsilon}  \C^{I\triangle J}\dashrightarrow (\Cx)^{I\triangle J}\xrightarrow{\check \rho\circ \check {\gamma}_{J,I} }\check T_\C,\]
where $\epsilon$ is identity on the factors corresponding to $I\setminus J$, and negative of the identity on the factors corresponding to $J\setminus I$. The is precisely the definition of $\phi_{J,I}$ in \Cref{YM(W)}, this proves the proposition.
\end{proof}

\begin{exmp}
	Consider the case when 
	\[T=\{(a,a,\dots,a)\in T^n:a\in S^1\}.\]
	A coordinate for $\check T_\C$ (which is a quotient of $(\Cx_b)^n$) would be $z=\prod_{i=1}^n b_i$. On the other hand, we have
	\[\mu_{\check T'_\C}^{-1}(0)=\{(b,a):b_ia_i=b_ja_j\  \forall\  i,j\}\]
	Hence we can use $\zeta=\zeta_1=\zeta_2=\cdots=\zeta_n$ as the coordinate for the cotangent fibre. The change of coordinate is now given as
	\[(\zeta^{|J|-|I|}\cdot z,\zeta)_I\sim (z,\zeta)_J.\]
\end{exmp}

\begin{exmp}
If $T$ is the trivial group, then $X^!=[T^*\C^n\sslash (\Cx)^n]$. In this case 
\begin{align*}
	\mathcal{U_I}&=[Z_I\sslash (\Cx)^n]\\
	&=\{(b,a)\in \C^n_b\times \C^n_a: b_i=0\iff i \in I\text{ and }a_i=0\iff i \not \in I\}/(\Cx)^n \\
	&\cong \pt.
\end{align*}
So $X^{\!}$ contains a disjoint union of $2^n$ points.
\end{exmp}

\printbibliography

\end{document}